\documentclass[11pt]{article}
\usepackage[letterpaper, margin=1in]{geometry}
\newif\iflong
\longtrue

\usepackage{fullpage}
\usepackage[english]{babel}
\usepackage{amssymb}
\usepackage[leqno]{amsmath}
\makeatletter
\newcommand{\leqnomode}{\tagsleft@true\let\veqno\@@leqno}
\newcommand{\reqnomode}{\tagsleft@false\let\veqno\@@eqno}
\makeatother

\usepackage{nicefrac}
\usepackage{mathabx}
\usepackage{mathtools}
\usepackage{csquotes}
\usepackage[all]{nowidow}
\usepackage{amsthm}
\usepackage[mathcal,mathscr]{eucal}
\usepackage{epsfig,graphicx,graphics,color}
\usepackage{enumitem}
\usepackage[sort,nocompress]{cite}
\usepackage[hypertexnames=false]{hyperref}

\hypersetup{
    colorlinks=true,            % false: boxed links; true: colored links
    linkcolor=[rgb]{0,0,.8},    % color of internal links (change box color with linkbordercolor)
    citecolor=[rgb]{0.85,0,0},  % color of links to bibliography
    filecolor=magenta,          % color of file links
    urlcolor=cyan,              % color of external links
    linktocpage=true
}

\usepackage{bbm}
\usepackage{varwidth}

\newcommand{\seeinappendix}{}
\newcommand{\thmrestate}[3]{
\begingroup
\def\thetheorem{\ref*{#1}}
\begin{#3}%[\textbf{restated}]
  #2
\end{#3}
\addtocounter{theorem}{-1}
\endgroup
}

\usepackage{algorithm}
\usepackage[noend]{algpseudocode}

\makeatletter
\renewcommand{\ALG@beginalgorithmic}{\small}
\makeatother

\newcounter{algsubstate}

\newenvironment{algsubstates}
  {\setcounter{algsubstate}{0}%
   \renewcommand{\State}{%
     \refstepcounter{algsubstate}%
     \Statex {\hspace{\algorithmicindent}\footnotesize\alph{algsubstate}:}\space}}
  {}

\usepackage[nameinlink]{cleveref}
\theoremstyle{plain}
\newtheorem{theorem}{Theorem}
\newtheorem{lemma}[theorem]{Lemma}
\newtheorem{corollary}[theorem]{Corollary}
\newtheorem{claim}[theorem]{Claim}
\newtheorem{remark}{Remark}

\crefname{empty}{}{}
\creflabelformat{empty}{#2(#1)#3}
\crefname{type}{type}{types}
\crefname{type1}{}{}
\creflabelformat{type1}{#2type (i)#3}
\crefname{type2}{}{}
\creflabelformat{type2}{#2type(ii)#3}
\creflabelformat{type}{#2#1#3}
\crefname{line}{step}{steps}
\crefname{algsubstate}{step}{steps}
\crefname{appendixa}{}{}
\creflabelformat{appendixa}{#2Appendix A#3}
\crefname{appendixb}{}{}
\creflabelformat{appendixb}{#2Appendix B#3}

\newcommand{\midbar}{~\middle|~}

\renewcommand{\O}{\mathcal{O}}

\newcommand{\TP}{\mathsf{TP}}
\newcommand{\T}{\mathsf{T}}
\renewcommand{\P}{\mathsf{P}}

\newcommand{\eps}{\varepsilon}
\newcommand{\Eps}{\mathcal{E}}
\newcommand{\varT}{\mathcal{T}}
\newcommand{\varL}{\mathcal{L}}

\newcommand{\Pj}[1]{P_{#1\text{-join}}^{\uparrow}}

\DeclareMathOperator\MV{MV}
\DeclareMathOperator\HK{HK}
\DeclareMathOperator\SG{CG}
\newcommand{\PHK}{P_{\HK}^{\MV}}
\newcommand{\PSG}{P^{\SG}}
\newcommand{\B}{\mathcal{B}}
\newcommand{\imin}{{i_{\min}}}
\newcommand{\imax}{{i_{\max}}}
\newcommand{\ddelta}{\dot{\delta}}

\renewcommand{\geq}{\geqslant}
\renewcommand{\leq}{\leqslant}

\DeclareMathOperator\supp{supp}
\DeclareMathOperator\spa{span}
\DeclareMathOperator\comp{comp}

\DeclareMathOperator\odd{odd}
\DeclareMathOperator\argmin{argmin}
\DeclareMathOperator\poly{poly}
\DeclareMathOperator\OPT{OPT}
\DeclareMathOperator\LP{LP}

\DeclarePairedDelimiter{\floor}{\lfloor}{\rfloor}

\newcommand\restr[2]{{% we make the whole thing an ordinary symbol
  \left.\kern-\nulldelimiterspace % automatically resize the bar with \right
  #1 % the function
  \vphantom{\big|} % pretend it's a little taller at normal size
  \right|_{#2} % this is the delimiter
  }}
\def\final{0}  % set this to 1 to get a comment-free version
\ifnum\final=0  %namely if we allow comments in the output
\newcommand{\knote}[1]{{\color{red}[{\tiny Krist\'of: \bf #1}]\marginpar{\color{red}*}}}
\newcommand{\rnote}[1]{{\color{blue}[{\tiny Roland: \bf #1}]\marginpar{\color{blue}*}}}
\newcommand{\mnote}[1]{{\color{red}[{\tiny Matthias: \bf #1}]\marginpar{\color{red}*}}}
\else % in this case [final=1] we don't want any comments to show
\newcommand{\knote}[1]{}
\newcommand{\rnote}[1]{}
\newcommand{\mnote}[1]{}
\fi

\title{A $\nicefrac32$-Approximation for the Metric Many-visits Path TSP\thanks{Supported by DAAD with funds of the Bundesministerium f{\"u}r Bildung und Forschung (BMBF) and by DFG project MN 59/4-1.}}

\author{Krist{\'o}f B{\'e}rczi\thanks{MTA-ELTE Egerv\'ary Research Group, Department of Operations Research, E{\"o}tv{\"o}s Lor{\'a}nd University, Budapest, Hungary. \texttt{berkri@cs.elte.hu}.} 
  \and Matthias Mnich\thanks{TU Hamburg, Hamburg, Germany. \texttt{matthias.mnich@tuhh.de}.}
  \and Roland Vincze\thanks{TU Hamburg, Hamburg, Germany. \texttt{roland.vincze@tuhh.de}.}
}

\begin{document}
\date{}
\maketitle

\begin{abstract}
  In the {\sc Many-visits Path TSP}, we are given a set of $n$ cities along with their pairwise distances (or cost) $c(uv)$, and moreover each city $v$ comes with an associated positive integer request $r(v)$.
  The goal is to find a minimum-cost path, starting at city $s$ and ending at city~$t$, that visits each city $v$ exactly $r(v)$ times.
  
  We present a $\nicefrac32$-approximation algorithm for the metric {\sc Many-visits Path TSP}, that runs in time polynomial in $n$ and \emph{poly-logarithmic} in the requests $r(v)$.
  Our algorithm can be seen as a far-reaching generalization of the $\nicefrac32$-approximation algorithm for {\sc Path TSP} by Zenklusen (SODA 2019), which answered a long-standing open problem by providing an efficient algorithm which matches the approximation guarantee of Christofides' algorithm from 1976 for metric TSP.
  
  One of the key components of our approach is a polynomial-time algorithm to compute a connected, degree bounded multigraph of minimum cost.
  We tackle this problem by generalizing a fundamental result of Kir{\'a}ly, Lau and Singh (Combinatorica, 2012) on the {\sc Minimum Bounded Degree Matroid Basis} problem, and devise such an algorithm for general polymatroids, even allowing element multiplicities.

  Our result directly yields a $\nicefrac32$-approximation to the metric {\sc Many-visits TSP}, as well as a $\nicefrac32$-approximation for the problem of scheduling classes of jobs with sequence-dependent setup times on a single machine so as to minimize the makespan.

\end{abstract}

\raisebox{-64ex}[0pt][0pt]{\hspace{88ex}\includegraphics[scale=0.45]{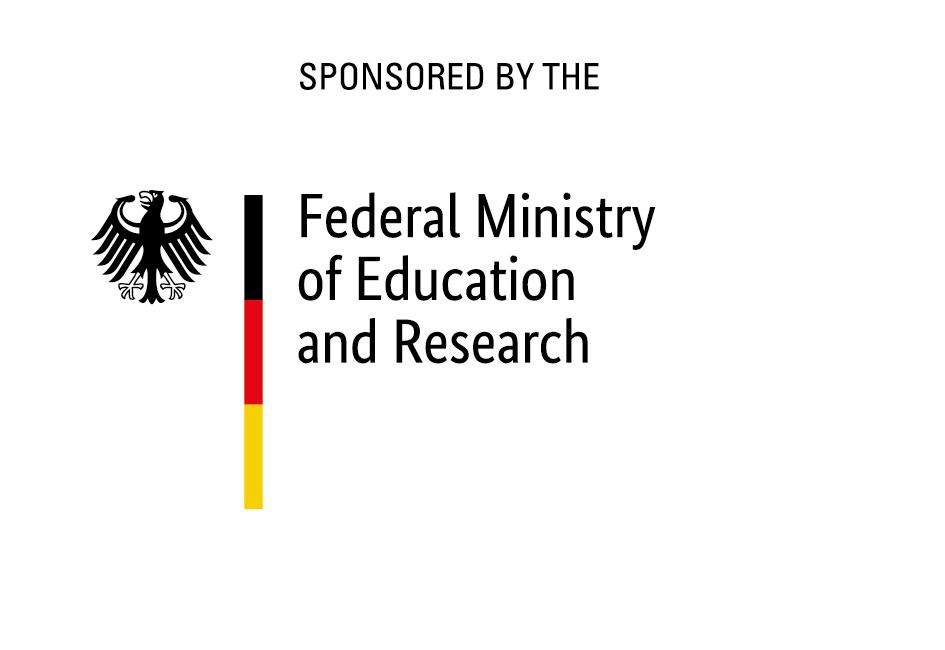}}

\thispagestyle{empty}
\setcounter{page}{1}
\clearpage
\pagebreak

\section{Introduction}
\label{sec:introduction}
The traveling salesman problem (TSP) is one of the cornerstones of combinatorial optimization.
Given a set $V$ of $n$ cities with non-negative costs $c(uv)$ for each cities $u$ and $v$, the objective is to find a minimum cost closed walk visiting each city.
TSP is well-known to be $\mathsf{NP}$-hard even in the case of metric costs, i.e. when the cost function $c$ satisfies the triangle inequality.
For metric costs, the best known approximation ratio that can be obtained in polynomial time is~$\nicefrac32$, discovered independently by Christofides~\cite{Christofides1976} and Serdyukov~\cite{Serdyukov1978}.

In the traveling salesman path problem, or {\sc Path TSP}, two distinguished vertices $s$ and $t$ are given, and the goal is to find a minimum cost walk from $s$ to $t$ visiting each city.
Approximating the metric {\sc Path TSP} has a long history, from the first $\nicefrac53$-approximation by Hoogeveen~\cite{Hoogeveen1991}, through subsequent improvements~\cite{AnKS2015, Sebo2013, Vygen2016, GottschalkVygen2018, SebovanZuylen2016} to the recent breakthroughs.
The latest results eventually closed the gap between the metric TSP and the metric {\sc Path TSP}: Traub and Vygen~\cite{TraubVygen2019} provided a $(\nicefrac{3}{2}+\eps)$-approximation for any $\varepsilon > 0$, Zenklusen~\cite{Zenklusen2019} provided a $\nicefrac{3}{2}$-approximation and finally the three authors showed a reduction from the {\sc Path TSP} to the TSP~\cite{TraubVZ2020}.

We consider a far-reaching generalization of the metric {\sc Path TSP}, the metric {\sc Many-visits Path TSP}, where in addition to the costs $c$ on the edges, a requirement $r(v)$ is given for each city~$v$.
The aim is to find a minimum cost walk from $s$ to $t$ that visits each city $v$ exactly $r(v)$ times.
The cycle version of this problem, where $s=t$, is known as {\sc Many-visits TSP} and was first considered in 1966 by Rothkopf~\cite{Rothkopf1966}.
Psaraftis~\cite{Psaraftis1980} proposed a dynamic programming approach that solves the problem in time $(\nicefrac{r}{n})^n$ for $r = \sum_{v\in V}r(v)$.
Later, Cosmadakis and Papadimitriou~\cite{CosmadakisPapadimitriou1984} gave the first algorithm for {\sc Many-visits TSP} with logarithmic dependence on $r$, though the space and time requirements of their algorithm were still superexponential in $n$.
Recently, Berger et~al.~\cite{BergerKMV2020} simultaneously improved the run time to $2^{\O(n)}\cdot \log r$ and reduced the space complexity to polynomial.
(The algorithm by Berger et al.~\cite{BergerKMV2020} can be slightly modified to solve the path version as well.)
Lately, Kowalik et al.~\cite{KowalikLNSW2020} made further fine-grained time complexity improvements.
To the best of our knowledge, no constant-factor approximation algorithms for the metric {\sc Many-visits TSP}\footnote{At the Hausdorff Workshop on Combinatorial Optimization in 2018, Rico Zenklusen brought up the topic of approximation algorithms for the metric version of {\sc Many-visits TSP} in the context of iterative relaxation techniques; he suggested an approach to obtain a 1.5-approximation, which is unpublished.} or metric {\sc Many-visits Path TSP} are currently known.

Besides being of scientific interest in itself, the {\sc Many-visits Path TSP} can be used for modeling various problems.
The \textit{aircraft sequencing problem} or \textit{aircraft landing problem} is one of the most referred applications in the literature~\cite{Psaraftis1980,Bianco1999,Beasley2000,Lieder2015}, where the goal is to find a schedule of departing and/or landing airplanes that minimizes an objective function and satisfies certain constraints.
The aircraft are categorized into a small number of classes, and for each pair of classes a non-negative lower bound is given denoting the minimum amount of time needed to pass between the take off/landing of two planes from the given classes. The problem can be embedded in the {\sc Many-visits Path TSP} model by considering the classes to be cities and the separation times to be costs between them, while the number of airplanes in a class corresponds to the number of visits of a city.

As another illustrious example, the {\sc Many-visits Path TSP} is equivalent to the high-mul\-ti\-pli\-ci\-ty job scheduling problem $1|HM, s_{ij}, p_j|\sum C_j$, where each class $j$ of jobs has a processing time~$p_j$ and there is a setup time $s_{ij}$ between processing two jobs of different classes.
There is only a handful of constant-factor approximation algorithms for scheduling problems with setup times~\cite{AllahverdiNCK2008}, see for example the results of Jansen et al.~\cite{JansenKMR2019} or Deppert and Jansen~\cite{DeppertJansen2019} that consider sequence-independent batch setup times, or van der Veen et al.~\cite{vanderVeenWZ1998} that considers sequence-dependent setup times with a special structure.
An approximation algorithm for the {\sc Many-visits Path TSP} would further extend the list of such results. 

A different kind of application comes from geometric approximation.
Recently, Kozma and M\"omke provided an EPTAS for the {\sc Maximum Scatter TSP}~\cite{KozmaMomke2017}. 
Their approach involved grouping certain input points together and thus reducing the input size.
The reduced problem is exactly the {\sc Many-visits TSP}.
The same problem arises as a subproblem in the fixed-parameter algorithm for the {\sc Hamiltonian Cycle} problem on graphs with bounded neighborhood diversity~\cite{Lampis2012}.

Our work relies on a polymatroidal optimization problem with degree constraints.
An illustrious example of such problem is the {\sc Minimum Bounded Degree Spanning Tree} problem, where the goal is to find a minimum cost spanning tree in a graph with lower and upper bounds on the degree of each vertex.
Checking feasibility of a degree-bounded spanning tree contains the $\mathsf{NP}$-hard {\sc Hamiltonian Path} problem, and several algorithms were given that were balancing between the cost of the spanning tree and the violation of the degree bounds~\cite{ChaudhuriRRT2009, ChaudhuriRRT2009a, FurerRaghavachari1994, Goemans2006, KonemannRavi2003, KonemannRavi2002}.
Based on an iterative rounding approach~\cite{Jain2001} combined with a relaxation step, Singh and Lau~\cite{SinghLau2015} provided a polynomial-time algorithm that finds a spanning tree of cost at most the optimum value violating each degree bound by at most~1.
Kir{\'a}ly et al.~\cite{KiralyLS2012} later showed that similar results can be obtained for the more general {\sc Minimum Bounded Degree Matroid Basis Problem}.

\subsection*{Our results}
In this paper we provide the first efficient constant-factor approximation algorithm for the metric {\sc Many-visits Path TSP}.
Formally, a graph $G=(V,E)$ is given with a positive integer~$r(v)$ for each $v\in V$, and a non-negative cost $c(uv)$ for every pair of vertices $u, v$; finally, a departure city~$s$ and an arrival city $t$ are specified.
We seek a minimum cost $s$-$t$-walk that visits each city~$v$ exactly $r(v)$ times, where leaving city $s$ as well as arriving to city $t$ counts as one visit.

The cost function $c:E\rightarrow\mathbbm{R}_{\geq 0}$ is assumed to be metric.
Besides the triangle inequality $c(uw) \leq c(uv) + c(vw)$ for every triplet $u, v, w$ this implies that the cost of a self-loop $c(vv)$ at vertex $v$ is at most the cost of leaving city $v$ to any other city $u$ and returning, that is:
\begin{equation*}
  c(vv) \leq 2 \cdot \! \min_{u \in V-v} c(uv) \qquad \text{for all } v \in V \enspace .
\end{equation*}
The assumption of metric costs is necessary, as the TSP, and therefore the {\sc Many-visits TSP}, does not admit any non-trivial approximation for unrestricted cost functions assuming that $\mathsf{P} \neq\mathsf{NP}$ (see e.g. Theorem~6.13 in the book of Garey and Johnson~\cite{GareyJohnson1979}).

We start with a simple approximation idea, that leads to a constant factor approximation in strongly polynomial time:

\newcommand{\thmpathsimple}{
  There is a polynomial-time $\nicefrac52$-approximation for the metric {\sc Ma\-ny-visits Path TSP}, that runs in time polynomial in $n$ and $\log r$.
}

\begin{theorem}
\label{thm:path25}
  \thmpathsimple
\end{theorem}

\iflong \else
\emph{\textbf{Note:} Due to length constraints, the $\nicefrac52$-approximation result is deferred to \Cref{appendixa}.} \vspace{1em}
\fi

The approximation factor $\nicefrac52$ in \Cref{thm:path25} still leaves a gap to the best-known factor~$\nicefrac32$ for the metric {\sc Path TSP}, which is due to Zenklusen~\cite{Zenklusen2019}.
His recent $\nicefrac32$-approximation for the metric {\sc Path TSP} uses a Christofides-Serdyukov-like construction that combines a spanning tree and a matching, with the key difference that it calculates a constrained spanning tree in order to bound the costs of the tree and the matching by $\nicefrac32$ times the optimal value.

Our main algorithmic result matches this approximation ratio for the metric {\sc Many-visits Path TSP}.

\newcommand{\thmpath}{
  There is a polynomial-time $\nicefrac32$-approximation for the metric {\sc Many-visits Path TSP}.
  The algorithm runs in time polynomial in $n$ and $\log r$.
}

\begin{theorem}
\label{thm:path15}
  \thmpath
\end{theorem}

As a direct consequence of \Cref{thm:path15}, we obtain the following:

\newcommand{\thmtsp}{
  There is a $\nicefrac32$-approximation for the metric {\sc Many-visits TSP} that runs in time polynomial in $n$ and $\log r$.
}

\begin{corollary}
\label{thm:tsp15}
  \thmtsp
\end{corollary}

Our approach follows the main steps of Zenklusen's work~\cite{Zenklusen2019}.
However, the presence of requests $r(v)$ makes the problem significantly more difficult and several new ideas are needed to design an algorithm which returns a tour with the correct number of visits and still runs in polynomial time.
%which is polynomial in $n$ and poly-logarithmic in the $r(v)$ values.
For instance, whereas the backbone of both Christofides and Zenklusen's algorithm is a spanning tree (with certain properties), the possibly exponentially large number of (parallel) edges in a many-visits TSP solution requires us to work with a structure that is more general than spanning trees.
We therefore consider the problem of finding a minimum cost connected multigraph with lower bounds~$\rho$ on the degree of vertices, and lower and upper bounds $L$ and $U$, respectively, on the number of occurrences of the edges.
We call this task the {\sc Minimum Bounded Degree Connected Multigraph with Edge Bounds} problem, and show the following:

\newcommand{\thmbdmultigraph}{
  There is an algorithm for the {\sc Minimum Bounded Degree Connected Multigraph with Edge Bounds} problem that, in time polynomial in $n$ and $\log \sum_v \rho(v)$, returns a connected multigraph~$T$ with $\nicefrac{\rho(V)}{2}$ edges, where each vertex $v$ has degree at least $\rho(v)-1$ and the cost of $T$ is at most the cost of $\min\{c^\T x \mid x \in \PSG(\rho, L, U)\}$, where
}

\begin{theorem}
\label{thm:bdmultigraph}
  % text of theorem is defined above
  \thmbdmultigraph
  \begin{equation}
  \label{eq:pcg_general}
    \PSG(\rho, L, U) := \left\{ x \in \mathbbm{R}^E_{\geq 0} \midbar \begin{array}{ll}
    \supp(x) \text{ is connected} \\
    x(E) = \nicefrac{\sum_v\rho(v)}{2} \\
    x(\ddelta(v)) \geq \rho(v)  \qquad \qquad \quad \qquad \forall v \in V \\
    L(vw)\leq x(vw) \leq U(vw) \qquad \forall v,w \in V
    \end{array} \right\} \enspace .
  \end{equation}

\end{theorem}

Note that an optimal solution $x^*$ to the {\sc Minimum Bounded Degree Connected Multigraph with Edge Bounds} problem is a minimum cost integral vector from the polytope~$\PSG$.
We use the result of \Cref{thm:bdmultigraph} to obtain a multigraph that serves a key role in our approximation algorithm for the metric {\sc Many-visits Path TSP}; the values $\rho$, $L$ and $U$ depend on the  instance and the details are given in \Cref{sec:32approximation}.

The {\sc Minimum Bounded Degree Connected Multigraph with Edge Bounds} problem shows a lot of similarities to the {\sc Minimum Bounded Degree Spanning Tree} problem.
However, neither the result of Singh and Lau~\cite{SinghLau2015} nor the more general approach by Kir{\'a}ly et al.~\cite{KiralyLS2012} applies  to our setting, due to the presence of parallel edges and self-loops in a multigraph.

One of our key contributions is therefore an extension of the result of Kir\'aly et al.~\cite{KiralyLS2012} to generalized polymatroids, which might be of independent combinatorial interest. 
Formally, the {\sc Bounded Degree g-po\-ly\-mat\-ro\-id Element with Multiplicities} problem takes as input a g-polymatroid $Q(p,b)$ defined by a paramodular pair $p,b:2^S\rightarrow\mathbbm{R}$, a cost function $c:S \rightarrow \mathbbm{R}$, a hypergraph $H=(S, \Eps)$ with lower and upper bounds $f,g:\Eps\rightarrow\mathbbm{Z}_{\geq 0}$ and multiplicity vectors $ m_\eps:S\rightarrow\mathbbm{Z}_{\geq0}$ for $\eps\in\Eps$ satisfying $m_\eps(s)=0$ for $s\in S-\eps$.
The objective is to find a minimum-cost integral element $x$ of $Q(p,b)$ such that $f(\eps) \leq \sum_{s\in \eps} m_\eps(s) x(s) \leq g(\eps)$ for each $\eps \in \Eps$.
We give a polynomial-time algorithm for finding a solution of cost at most the optimum value with bounds on the violations of the degree prescriptions.

\newcommand{\thmmatroidtwoside}{
  There is an algorithm for the {\sc Bounded Degree g-polymatroid Element with Multiplicities} problem which returns an integral element $x$ of $Q(p,b)$ of cost at most the optimum value such that $f(\eps)- 2\Delta+1 \leq \sum_{s\in \eps} m_\eps(s) x(s) \leq g(\eps)+2\Delta-1$ for each $\eps\in\Eps$, where $\Delta=\max_{s\in S}\left\{\sum_{\eps\in\Eps:s\in \eps} m_\eps(s)\right\}$.
  The run time of the algorithm is polynomial in $n$ and $\log \sum_\eps \left(f(\eps)+g(\eps)\right)$.
}

\begin{theorem}
\label{thm:matroid1}
  \thmmatroidtwoside
\end{theorem}

When only lower bounds (or only upper bounds) are present, we call the problem {\sc Lower (Upper) Bounded Degree g-polymatroid Element with Multiplicities}.
Similarly to Kir{\'a}ly et al.~\cite{KiralyLS2012}, we obtain an improved bound on the degree violations when only lower or upper bounds are present:~\footnote
  {
  The results in \Cref{thm:matroid1}, \Cref{thm:matroid2} and \Cref{thm:tsp15} appeared in an unpublished work~\cite{BercziBMV2019} by a superset of the authors.
  In order to make the paper self-contained, we include all the details and proofs in this paper as well.
  } 

\newcommand{\thmmatroidoneside}{
  There is an algorithm for {\sc Lower Bounded Degree g-polymatroid Element with Multiplicities} which returns an integral element $x$ of $Q(p,b)$ of cost at most the optimum value such that $f(\eps)- \Delta+1 \leq \sum_{s\in \eps} m_\eps(s) x(s)$ for each $\eps\in\Eps$.
  An analogous result holds for {\sc Upper Bounded Degree g-polymatroid Element}, where $\sum_{s\in \eps} m_\eps(s) x(s) \leq g(\eps) + \Delta - 1$.
  The run time of these algorithms is polynomial in $n$ and $\log \sum_\eps f(\eps)$ or $\log \sum_\eps g(\eps)$, respectively.
  }

\begin{theorem}
\label{thm:matroid2}
  \thmmatroidoneside
\end{theorem}

\section{Preliminaries}
\label{sec:pre}

\paragraph{Basic notation.}
Throughout the paper, we let $G=(V,E)$ be a finite, undirected complete graph on $n$ vertices, whose edge set $E$ also contains a self-loop at every vertex $v\in V$.
For a subset $F\subseteq E$ of edges, the \emph{set of vertices covered by $F$} is denoted by $V(F)$.
The \emph{number of connected components} of the graph $(V(F),F)$ is denoted by $\comp(F)$.
For a subset $X\subseteq V$ of vertices, the \emph{set of edges spanned by~$X$} is denoted by $E(X)$.
Given a multiset $F$ of edges (that is, $F$ might contain several copies of the same edge), the multiset of edges leaving the vertex set $C \subseteq V(F)$ is denoted by $\delta_F(C)$.
Similarly, denote the multiset of regular edges (i.e. excluding self-loops) in $F$ incident to a vertex $v\in V$ is denoted by~$\delta_F(v)$.
Denote the multiset of all edges (i.e. including self-loops) in $F$ incident to a vertex $v \in V$ by ~$\ddelta_F(v)$, then the \emph{degree} of $v$ in $F$ is denoted by $\deg_F(v) := |\ddelta_F(v)|$, where every copy of the self-loop at $v$ in $F$ is counted twice. 
We will omit the subscript when $F$ contains all the edges of $G$, that is, $F=E$. 
For a vector $x\in\mathbbm{R}^{E}$, we denote the sum of the $x$-values on the edges incident to~$v$ by~$x(\ddelta(v))$.
Note that the $x$-value of the self-loop at $v$ is counted twice in~$x(\ddelta(v))$.
Let us denote the set of edges between two disjoint vertex sets $A$ and $B$ by $\delta(A, B)$.
Given two graphs or multigraphs~$H_1,H_2$ on the same vertex set, $H_1+H_2$ denotes the multigraph obtained by taking the union of the edge sets of~$H_1$ and~$H_2$.

Given a vector $x\in\mathbbm{R}^{S}$ and a set $Z\subseteq S$, we use $x(Z)=\sum_{s\in Z} x(s)$.
The \emph{lower integer part} of~$x$ is denoted by $\floor{x}$, so $\floor{x} (s)=\floor{x(s)}$ for every $s\in S$.
This notation extends to sets, so by $\floor{x}(Z) $ we mean $\sum_{s \in Z} \floor{x}(s)$.
The \emph{support of $x$} is denoted by $\supp(x)$, that is, $\supp(x)=\{s\in S\mid x(s)\neq 0\}$.
The \emph{difference of set $B$ from set $A$} is denoted by $A-B=\{s\in A \mid s\notin B\}$.
We denote a single-element set $\{s\}$ by $s$, and with a slight abuse of notation, we write $A-s$ to indicate $A- \{s\}$.
Let us denote the {\it symmetric difference} of two sets $A$ and $B$ by $A \triangle B := (A - B) \cup (B - A)$ and the \emph{characteristic vector} of a set $A$ by $\chi_A$.

For a collection $\varT$ of subsets of $S$, we call $\mathcal{L} \subseteq \varT$ an \emph{independent laminar system} if for any pair $X, Y \in \mathcal{L}$: (i) they do not properly intersect, i.e. either $X \subseteq Y$, $Y \subseteq X$ or $X \cap Y = \emptyset$, and (ii) the characteristic vectors $\chi_Z$ of the sets $Z \in \mathcal{L}$ are independent over the real numbers.
A \emph{maximal} independent laminar system $\mathcal{L}$ with respect to $\varT$ is an independent laminar system in $\varT$ such that for any $Y \in \varT-\mathcal{L}$ the system $\mathcal{L} \cup \{Y\}$ is not independent laminar.
In other words, if we include any set~$Y$ from $\varT-\mathcal{L}$, it will intersect at least one set $Y$ from $\mathcal{L}$, or~$\chi_Y$ can be given as a linear combination of $\{ \chi_Z\mid Z \in \mathcal{L} \}$. 
Given a laminar system~$\mathcal{L}$ and a set $X\subseteq S$, the set of maximal members of $\mathcal{L}$ lying inside $X$ is denoted by $\mathcal{L}^{\max}(X)$, that is, $\mathcal{L}^{\max}(X)=\{Y\in\mathcal{L}\mid Y\subsetneq X,\ \not\exists Y'\in\mathcal{L}\ \text{s.t.}\ Y\subsetneq Y'\subsetneq X\}$.

\paragraph{Many-visits Path TSP.}

Recall that in the {\sc Many-visits Path TSP}, we seek for a minimum cost $s$-$t$-walk $P$ such that $P$ visits each vertex $v \in V$ exactly $r(v)$ times.
Let $r(V) = \sum_{v\in V}r(v)$.
The sequence of the edges of $P$ has length $r(V)-1$, which is exponential in the size of the input, as the values $r(v)$ are stored using $\log r(V)$ space.
For this reason, instead of explicitly listing the edges in a walk (or tour) we always consider \emph{compact representations} of the solution and the multigraphs that arise in our algorithms.
That is, rather than storing an $(r(V)-1)$-long sequence of edges, for every edge $e$ we store its multiplicity $z(e)$ in the solution.
As there are at most~$n^2$ different edges in the solution each having multiplicity at most $\max_{v \in V} r(v)$, the space needed to store a feasible solution is $\O(n^2\log r(V))$.
Therefore, a vector $z \in \mathbbm{Z}_{\geq 0}^{E}$ represents a feasible tour if $\supp(z)$ is a connected subgraph of $G$ and $\deg_z(v)=2\cdot r(v)$ holds for all $v\in V - \{s,t\}$ and $\deg_z(v)=2 \cdot r(v) -1$ for $v \in \{s,t\}$.
(Note that each self-loop $vv$ contributes~$2$ in the value $\deg(v) = |\ddelta(v)|$.)

Denote by $\P^\star_{c,r,s,t}$ an optimal solution for an instance $(G,c,r,s,t)$ of the {\sc Many-visits Path TSP}.
Let us denote by~$\P^\star_{c,1,s,t}$ an optimal solution for the single-visit counterpart of the problem, i.e. when $r(v)=1$ for each $v\in V$. 
Relaxing the connectivity requirement for solutions of the {\sc Many-visits Path TSP} yields Hitchcock's transportation problem~\cite{Hitchcock1941}, where supply and demand vertices $\{a_v\}_{v\in V}$ and $\{b_v\}_{v\in V}$ are given.
The supplies for $v \in V - s$ are then defined by $r(v)$, the supply of $s$ by $r(s)-1$; the demand of each vertex $v \in V - t$ by $r(v)$ and the demand of $t$ by $r(t)-1$.
Finally, by setting the transportation costs between $a_u$ and $b_v$ as $c(uv)$, the objective is to fulfill the supply and demand requirements by transporting goods from vertices $\{a_v\}_{v\in V}$ to vertices in $\{b_v\}_{v\in V}$, while keeping the total cost minimal.
The transportation problem is solvable in polynomial time using a minimum cost flow algorithm~\cite{EdmondsKarp1970} and we denote an optimal solution by~$\TP^\star_{c,r,s,t}$, where $s$ and $t$ denote the special vertices with decreased supply and demand value, respectively.

\newcommand{\lempathcycle}{
  Let $\TP^\star_{c,r,s,t}$ be an optimal solution to the Hitchcock transportation problem, where $supply(v)+demand(v)$ is odd for $v\in\{s,t\}$ and it is even otherwise.
  Then $\TP^\star_{c,r,s,t}$ can be decomposed into cycles and exactly one $s$-$t$-path.
}

\begin{lemma}
\label{lem:path_cycle}
  \lempathcycle
\end{lemma}

\newcommand{\prooflempathcycle}{
\begin{proof}
  Any solution $X$ to the transportation problem is essentially a multigraph that has an even degree for vertices $v \in V - \{s,t\}$, and an odd degree for $v \in \{s,t\}$.
  Hence, because of a parity argument, there has to be an $s$-$t$-path $U$ in $X$, possibly covering other vertices $W \subset V - \{s,t\}$.
  Vertices $w \in W$ have an even degree in $U$.
  Therefore, deleting the edges of $U$ from $X$, all vertices $v \in V$ will have an even degree in the modified multigraph $X'$.
  Thus~$X'$ can be decomposed into a union of (not necessarily distinct) cycles, and the lemma follows.
\end{proof}
}

\iflong \prooflempathcycle \else \seeinappendix \fi

\iflong \else
\emph{\textbf{Note:} From now on, the proofs are deferred to \Cref{appendixb}, unless stated otherwise.} \vspace{1em}
\fi 

The decomposition provided by the lemma is called a path-cycle representation.
Such a representation can be stored as a path $P_0$ and a collection $\mathcal{C}$ of pairs $(C, \mu_C)$, where each~$C$ is a simple closed walk (cycle) and $\mu_C$ is the corresponding integer denoting the number of copies of $C$.
Below we show that one can always calculate a path-cycle decomposition in polynomial time, and such a decomposition takes polynomial space.

\newcommand{\lemcompactpathcycle}{
  Let $\P_{c,r,s,t}$ be a many-visits TSP path with endpoints $s, t$, and $\TP_{c,r,s,t}$ be a transportation problem solution with special vertices $s, t$.
  There is a path-cycle representation of $\P_{c,r,s,t}$ and $\TP_{c,r,s,t}$, both of which take space polynomial in $n$ and $\log r(V)$, and can be computed in time polynomial in~$n$ and $\log r(V)$.
  }

\begin{lemma}
\label{lem:compact_path_cycle}
  \lemcompactpathcycle
\end{lemma}

\newcommand{\prooflemcompactpathcycle}{
\begin{proof}
  We first show the proof for a many-visits TSP path $\P_{c,r,s,t}$.
  Let us first add an edge $ts$ to $\P_{c,r,s,t}$, and denote the resulting multigraph by $T$.
  Observe that $T$ is a many-visits TSP tour with the same number of visits, since it is connected and the degree of every vertex $v$ in $T$ is $2 \cdot r(v)$.
  We can now use the procedure \textbf{ConvertToSequence} by Grigoriev and van de Klundert~\cite{Grigoriev2006}, which takes the edge multiplicities of $T$, denoted by $\{x_{uv}\}_{u,v \in V}$ as input, and outputs a collection $\mathcal{C}$ of pairs~$(C, \mu_C)$.
  Here, $C$ is a simple closed walk, and $\mu_C$ is the corresponding integer denoting the number of copies of the walk $C$ in $P$.
  Lastly, choose an arbitrary cycle $C$, such that $ts \in C$, and transform one copy of $C$ into a path as follows.
  Let $C_0:=C$, and remove the edge $ts$ from~$C_0$, resulting in an $s$-$t$-path $P_0$.
  Update $\mu_C := \mu_C-1$.
  Now $(P_0, \mathcal{C})$ is a compact path-cycle representation of $\P_{c,r,s,t}$.

  In every iteration, the procedure \textbf{ConvertToSequence} looks for a cycle $C$ and removes each of its occurrences from $\{x_{uv}\}_{u,v \in V}$.
  The procedure stops when $\{x_{uv}\}_{u,v \in V}$ represents a graph without edges.
  This demonstrates that the input need not represent a connected graph in the first place, as the edge removals possibly make it disconnected during the process.
  Note that the only structural difference between $\TP_{c,r,s,t}$ and $\P_{c,r,s,t}$ is that the underlying multigraph of $\TP_{c,r,s,t}$ might be disconnected.
  This means that the procedure \textbf{ConvertToSequence} can be applied to obtain a compact path-cycle representation of $\TP_{c,r,s,t}$ the same way as in the case of $\P_{c,r,s,t}$.

  Finally, the number of cycles $C$ in $\mathcal{C}$ can be bounded by $\O(n^2)$ (as removing all occurrences of a cycle $C$ sets at least one variable $x_{uv}$ to zero), and the algorithm has a time complexity of $\O(n^4)$~\cite{Grigoriev2006}.
  The edge insertion and deletion, and other graph operations during the process, can also be implemented efficiently.
  This concludes the proof.
\end{proof}
}

\iflong \prooflemcompactpathcycle \else \seeinappendix \fi
From now on, we assume that the path-cycle decompositions appearing in this paper are stored in space polynomial in $n$ and $\log r(V)$.

Let $(P_0, \mathcal{C})$ be a compact path-cycle representation of a many-visits TSP path $\P_{c,r,s,t}$.
One can obtain the explicit order of the vertices from $(P_0, \mathcal{C})$ the following way: traverse the $s$-$t$-path~$P_0$, and whenever a vertex~$u$ is reached for the first time, traverse~$\mu_{C}$ copies of every cycle~$C$ containing~$u$.
Note that while the size of~$\mathcal{C}$ is polynomial in $n$, the size of the explicit order of the vertices is exponential, hence the approaches presented in this paper consider symbolic rather than literal traversals of many-visits TSP paths and tours.

\iflong
\section{A Simple \texorpdfstring{$\nicefrac52$}{5/2}-Approximation for Metric Many-visits Path TSP}
\label{sec:simple52approximation}

In this section we give a simple $\nicefrac52$-approximation algorithm for the metric {\sc Many-visits Path TSP} that runs in polynomial time.
The algorithm is as follows:
\begin{algorithm}[h!]
  \caption{A polynomial-time $(\alpha+1)$-approximation for metric {\sc Many-visits Path TSP}.\label{alg:apx_tp_path}}
  \begin{algorithmic}[1]
    \Statex \textbf{Input:} A complete undirected graph $G=(V,E)$, costs $c:E\rightarrow\mathbbm{R}_{\geq 0}$ satisfying the triangle inequality, requests $r:V\rightarrow\mathbbm{Z}_{\geq 1}$, distinct vertices $s,t\in V$.
    \Statex \textbf{Output:} An $s$-$t$-path that visits each $v \in V$ exactly $r(v)$ times. 
    \State Calculate an $\alpha$-approximate solution $\P^\alpha_{c,1,s,t}$ for the single-visit metric {\sc Path TSP} instance $(G,c,1,s,t)$.\label{st:25_singletsp}      
    \State Calculate an optimal solution $\TP^\star_{c,r,s,t}$ for the corresponding transportation problem, together with a compact path-cycle decomposition $(P_0, \mathcal{C})$, where $\mathcal{C}$ is a collection of pairs~$(C, \mu_C)$. \label{st:25_transport}
    \State Let $P$ be the union of $\P^\alpha_{c,1,s,t}$ and $\mu_C$ copies of every cycle $C \in \mathcal{C}$. \label{st:25_path}
    \State Do shortcuts in $P$ and obtain a solution $P'$, such that $P'$ visits every city $v$ exactly $r(v)$ times (that is, $\deg_{P'}(v) = 2 \cdot r(v)$ for every vertex $v \in V - \{s,t\}$, and $\deg_{P'}(v) = 2 \cdot r(v) -1$ otherwise). \label{st:25_goodpath}
    \State \textbf{return} $P'$.
  \end{algorithmic}
\end{algorithm}

\thmrestate{thm:path25}{\thmpathsimple}{theorem}

\newcommand{\proofthmpathsimple}{
\begin{proof}%[Proof of \Cref{thm:path25}]
  The algorithm is presented as \Cref{alg:apx_tp_path}.
  Since~$\P^\alpha_{c,1,s,t}$ is connected, and~$P$ contains all the edges of $\P^\alpha_{c,1,s,t}$, $P$ is also connected.
  Let $(P_0, \mathcal{C})$ be the compact path-cycle decomposition of $\TP^\star_{c,r,s,t}$.
  The graph $P$ thus consists of $\P^\alpha_{c,1,s,t}$ and the cycles of $\mathcal{C}$. 
  The edges of $\P^\alpha_{c,1,s,t}$ contribute a degree of 1 in case of $s$ and $t$, and 2 for $v \in V - \{s,t\}$; the cycles of $\mathcal{C}$ contribute degrees of $2 \cdot r(v)-2$ for $v \in \{s,t\}$, and degrees of $2 \cdot r(v)$ or $2 \cdot r(v)-2$ for $v \in V - \{s,t\}$.
  Let us denote the latter set by $W$, matching the notation in the proof of \Cref{lem:path_cycle}.
  The total degree of $v$ in~$P$ is:
  \begin{align*}
  2 \cdot r(v) - 1 \quad & \text{ for } v \in \{s,t\}, \\
  2 \cdot r(v) \quad & \text{ for } v\in W, \text{ and} \\
  2 \cdot r(v) + 2 \quad & \text{ for the remaining vertices in } V - (W \cup \{s,t\}).
  \end{align*}
  As a direct consequence of the degrees and connectivity, $P$ is an open walk that starts in $s$, visits every vertex $v \in V$ either $r(v)$ or $r(v)+1$ times, and ends in $t$.
  Since the edge costs are metric, we can use shortcuts at the vertices $w \in V - (W \cup \{s,t\})$ to reduce their degrees by 2.
  We describe the procedure below.
  
\paragraph{Shortcutting.}
  At \Cref{st:25_path}, $(\P^\alpha_{c,1,s,t}, \mathcal{C})$ denotes the compact path-cycle representation of $P$. 
  Let us construct an auxiliary multigraph $A$ on the vertex set $V$ by taking the edges of $\P^\alpha_{c,1,s,t}$ and each cycle~$C$ from $\mathcal{C}$ exactly once.
  Note that parallel edges appear in~$A$ if and only if an edge appears in multiple distinct cycles, or in the path $\P^\alpha_{c,1,s,t}$ and at least one cycle~$C$.
  Due to the construction, $s$ and $t$ have odd degree, while every other vertex has an even degree in~$A$, which means that there exist an Eulerian trail in $A$.
  Moreover, there are $\O(n^2)$ cycles~\cite{Grigoriev2006}, hence the total number of edges in $A$ is $\O(n^3)$.
  Consequently, using Hierholzer's algorithm, we can compute an Eulerian trail $\eta$ in $A$ in $\O(n^3)$ time~\cite{Hierholzer1873,Fleischner1991}.
  The trail $\eta$ covers the edges of each cycle~$C$ once.
  Now an implicit order of the vertices in the many-visits TSP path~$P$ is the following.
  Traverse the vertices of the Eulerian trail $\eta$ in order.
  Every time a vertex $u$ appears the first time, traverse all cycles~$C$ that contain the vertex $\mu_{C}$ times.
  Denote this trail by~$\eta'$.
  It is easy to see that the sequence~$\eta'$ is a sequence of vertices that uses the edges of $\P^\alpha_{c,1,s,t}$ once and the edges of each cycle $C$ exactly $\mu_C$ times, meaning this is a feasible sequence of the vertices in the path~$P$.
  Moreover, the order itself takes polynomial space, as it is enough to store indices of $\O(n^3)$ vertices and $\O(n^2)$ cycles.
  
  Denote the surplus of visits of a vertex $w\in W$ by $\gamma(w) := \nicefrac{\deg_P(w)}{2} - r(w)$.
  In \Cref{st:25_goodpath}, we remove the last~$\gamma(w)$ occurrences of every vertex $w \in W$ from~$P$ by doing shortcuts: if an occurrence of~$w$ is preceded by~$u$ and superseded by~$v$ in $P$, replace the edges $uw$ and $wv$ by $uv$ in the sequence.
  This can be done by traversing the compact representation of $\eta'$ backwards, and removing the vertex $w$ from the last $\gamma(w)$ cycles $C^{(w)}_{r(w)-\gamma(w)+1}, \dots, C^{(w)}_{r(w)}$.
  As $\sum_w \gamma(w)$ can be bounded by~$\O(n)$, this operation makes $\O(n)$ new cycles, keeping the space required by the new sequence of vertices and cycles polynomial.
  Moreover, since the edge costs are metric, making shortcuts the way described above cannot increase the total cost of the edges in $P$.
  Finally, using a similar argument as in the algorithm of Christofides, the shortcutting does not make the trail disconnected.
  The resulting graph is therefore an $s$-$t$-walk $P'$ that visits every vertex~$v$ exactly $r(v)$ times, that is, a feasible solution for the instance $(G,c,r,s,t)$.  
  
  Note that by construction, $P$ is such that the surplus of visits $\gamma(w)$ equals to either $0$ or~$1$.
  However, the same shortcutting procedure is used in \Cref{alg:mvsttsp15} later in the paper, where $\gamma(w)$ can take higher values as well.

\paragraph{Costs and complexity.}
  The cost of the path $P$ constructed by \Cref{alg:apx_tp_path} equals to $c(P') \leq c(\P^\alpha_{c,1,s,t}) + c(\TP^\star_{c,r,s,t})$.
  Since $c(\TP^\star_{c,r,s,t})$ is an optimal solution to a relaxation of the {\sc Many-visits Path TSP}, its cost is a lower bound to the cost of the corresponding optimal solution,~$\P^\star_{c,r,s,t}$.
  Since the cost of $\P^\alpha_{c,1,s,t}$ is at most $\alpha$ times the cost of an optimal single-visit TSP path $\P^\star_{c,1,s,t}$, and  $c(\P^\star_{c,1,s,t}) \leq c(\P^\star_{c,r,s,t})$ holds for any $r$, \Cref{alg:apx_tp_path} provides an $(\alpha+1)$-approximation for the {\sc Many-visits Path TSP}.  
  Using Zenklusen's recent polynomial-time $\nicefrac32$-approximation algorithm on the single-visit metric {\sc Path TSP}~\cite{Zenklusen2019} in \Cref{st:25_singletsp} yields the approximation guarantee of $\nicefrac{5}{2}$ stated in the theorem.

  The transportation problem in \Cref{st:25_transport} can be solved in $\O(n^3\log n)$ operations using the approach of Orlin~\cite{Orlin1993} or its extension due to Kleinschmidt and Schannath~\cite{KleinschmidtSchannath1995}.
  \Cref{st:25_path} can also be performed in polynomial time~\cite{Grigoriev2006}, and the number of closed walks can be bounded by~$\O(n^2)$.
  Moreover, the total surplus of degrees in $P$ is at most $n-2$, therefore the number of operations performed during shortcutting in \Cref{st:25_goodpath} is also bounded by $\O(n)$.
  This proves that the algorithm has a polynomial time complexity.~\footnote
  {
  One can obtain a $\nicefrac52$-approximation for the metric {\sc Many-visits TSP} by simply running \Cref{alg:apx_tp_path} for every pair $(u, v) \in V \times V$ and setting $s=u$ and $t=v$, then choosing a solution whose cost together with the cost of the edge $uv$ is minimal. 
  However, \Cref{alg:apx_tp_path} can be simplified while maintaining the same approximation guarantee. %, i.e. provided an $\alpha$-approximation to the classical TSP (e.g. the Christofides-Serdyukov algorithm), it yields an $(\alpha+1)$-approximation for the many-visits counterpart.
  This approach appeared in the unpublished manuscript~\cite{BercziBMV2019} by a superset of the authors and has a simpler proof, as the algorithm does not involve making shortcuts.
  }

\end{proof}
}

\proofthmpathsimple

\begin{remark}
The TSP, as well as the {\sc Path TSP} can also be formulated for directed graphs, where the costs $c$ are asymmetric.
(Note that $c$ still has to satisfy the triangle inequality, which implies the following bound for the self-loops: $c(vv) \leq \max_{u \neq v} \left\{ c(vu) + c(uv) \right\}$.)
In a recent breakthrough, Svensson et al.~\cite{SvenssonEtAl2018} gave the first constant-factor approximation for the metric {\sc ATSP}.
In subsequent work, Traub and Vygen~\cite{TraubVygen2020} improved the constant factor to $22 + \eps$ for any $\eps > 0$.
Moreover, Feige and Singh~\cite{FeigeSingh2007} proved that an $\alpha$-approximation for the metric {\sc ATSP} yields a $(2\alpha + \eps)$-approximation for the metric {\sc Path-ATSP}, for any $\eps > 0$.
By combining these results with a suitable modification of \Cref{alg:apx_tp_path}, we can obtain a $(23 + \eps)$-approximation for the metric {\sc Many-visits ATSP}, and a $(45 + \eps)$-approximation for any $\eps > 0$ for the metric {\sc Many-visits Path-ATSP} in polynomial time.
\end{remark}

\fi

\section{A \texorpdfstring{$\nicefrac32$}{3/2}-Approximation for the Metric Many-visits Path TSP}
\label{sec:32approximation}

In this section we show how to obtain a $\nicefrac32$-approximation for the metric {\sc Many-visits Path TSP}.
Our approach follows the general strategy of Zenklusen~\cite{Zenklusen2019}, but we need to make several crucial modifications for the many-visits setting with exponentially large requests.
This means that instead of calculating a constrained spanning tree, we use the result in \Cref{thm:bdmultigraph} to obtain a connected multigraph $P$ with a sufficiently large number of edges.
Then compute a matching~$M$ so that all the degrees in $P+M$ have the correct parity, and the cost of $P+M$ is at most $\nicefrac32$ times the optimal cost.
In order to ensure this bound, we have to enforce certain restrictions on $P$, similarly to the computation of the spanning tree in~\cite{Zenklusen2019}.
However, as we will show, the many-visits setting leads to further challenges.

For technical reasons, from now on we assume that the two endpoints $s$ and $t$ are different.
Let us start by defining the Held-Karp relaxation of the {\sc Many-visits Path TSP} as $\min\{c^\T x\mid x \in \PHK\}$, where $\PHK$ denotes the following polytope:
\begin{equation}
\label{eq:phk}
\PHK := \left\{ x \in \mathbbm{R}^E_{\geq 0} \midbar \begin{array}{ll}
x(\delta(C)) \geq 2  & \forall C \subset V, C \neq \emptyset, |C \cap \{s, t\}| \in \{0, 2\} \\
x(\delta(C)) \geq 1  & \forall C \subseteq V, |C \cap \{s, t\}| = 1 \\
x(\ddelta(v)) = 2 \cdot r(v) & \forall v \in V - \{s, t\}\\
x(\ddelta(s)) = 2 \cdot r(s) - 1 \\
x(\ddelta(t)) = 2 \cdot r(t) - 1
\end{array} \right\}
\end{equation}
The \emph{$Q$-join polytope} (where $Q \subseteq V$ is of even cardinality) is defined as follows:
\begin{equation}
\label{eq:pqjoin}
\Pj{Q} := \left\{x \in \mathbbm{R}^E_{\geq 0} \midbar x(\delta(C)) \geq 1 ~~ \forall Q\text{-cut} ~ C \subset V \right\},
\end{equation}
where a \emph{$Q$-cut} is a set $C \subseteq V$ with $|C \cap Q|$ odd.

In the following we assume arbitrary but fixed parameters $c, r$ and denote the optimal many-visits TSP path by $\P^\star = \P^\star_{c,r,s,t}$.
Given a solution $y$ of the linear program $\min\{c^\T x\mid x \in \PHK\}$, the vector $\nicefrac{y}{2}$ is not necessarily in $\Pj{Q}$ for every even set $Q \subseteq V$.
Indeed, $y$ only needs to have a load of 1 on $s$-$t$-cuts, therefore $\nicefrac{y}{2}$ may violate some of the constraints of $\Pj{Q}$.
This means that calculating a minimum cost perfect matching on an arbitrary even set $Q \subseteq V$ might lead to a matching $M$ with higher cost than $\nicefrac{c(\P^\star)}{2}$.
Therefore, simply taking a solution $P$ provided by \Cref{thm:bdmultigraph} and a minimum cost matching $M$ on the vertices with degrees having incorrect parity, then applying shortcuts would not lead to a $\nicefrac32$-approximation.

To circumvent this problem, we would like to have a control over the vertices of $P$ that take part in the perfect matching phase of the algorithm.
Similarly to Zenklusen~\cite{Zenklusen2019}, we calculate a point $q$ that is feasible for the Held-Karp relaxation of the {\sc Many-visits Path TSP}, and that is only needed for the analysis of the algorithm.
Let $\odd(P)$ denote the vertices $v$ with an odd degree in $P$.
We need $P$ and $q$ to meet the following requirements:
\begin{enumerate}[label={(R\arabic*)}]
  \item $c(P) \leq c(\P^\star)$, \label{req_Pq_1}
  \item $c(q) \leq c(\P^\star)$, and \label{req_Pq_2}
  \item $q/2 \in \Pj{Q_P}$, where $Q_P := \odd(P) \triangle \{s,t\}$, \label{req_Pq_3} 
\end{enumerate}
where $c(q)$ stands for the cost of the vector $q$ with respect to $c$, that is, $c(q)=\sum_{e\in E}c(e)q(e)$.

Adding a shortest $Q_P$-join $J$ to the multigraph $P$ results in a multigraph $P'$ where every vertex $v \in V - \{s,t\}$ has an even degree at least $2 \cdot r(v)$, and every $v \in \{s,t\}$ has an odd degree at least $2 \cdot r(v)-1$.
Due to \ref{req_Pq_3}, the cost of the shortest $Q_P$-join~$J$ satisfies $c(J) \leq \nicefrac{c(q)}{2}$. 
Therefore, using Wolsey's analysis for Christofides' algorithm, the solution $P''$ obtained by taking the edges of $P'$ and applying shortcuts has cost at most $\nicefrac{3}{2} \cdot c(\P^\star)$.

Let $x^*$ be an optimal solution to the Held-Karp relaxation of the {\sc Many-visits Path TSP}:
\begin{equation}
\label{eqn:mvpathtsp_hkrelax}
  \min\{c^\T x\mid x \in \PHK\} \enspace .
\end{equation}

In order to obtain $P$ and $q$ that satisfy the conditions \ref{req_Pq_1}-\ref{req_Pq_3} above, we will calculate another solution $y \in \PHK$ with $c(y) \leq c(\P^\star)$, and set $q$ to be the midpoint between $x^*$ and~$y$, that is, $q = \nicefrac{x^*}{2} + \nicefrac{y}{2}$.
The construction of $P$ also depends on $y$, the details are given in \Cref{alg:mvsttsp15} and the reasoning in the proof of \Cref{thm:path15}.
Being a convex combination of two points in $\PHK$, $q$ is in $\PHK$ as well.
We would like to ensure the existence of a multigraph $P$ such that $q/2 \in \Pj{Q_P}$, therefore we need to construct~$y$ accordingly.

Let $Q \subseteq V$ be a set of even cardinality.
Recall the definition of $\Pj{Q}$ at \Cref{eq:pqjoin}, which requires that the load on $Q$-odd cuts is at least 1.
Since $q$ is in $\PHK$, $\nicefrac{q(\delta(C))}{2} \geq 1$ holds for any non-$s$-$t$-cuts, i.e. any cuts $C \subset V, C \neq \emptyset$ with $|C \cap \{s,t\}| \in \{0,2\}$.
However, for $s$-$t$-cuts, the property $y \in \PHK$ only implies $y(\delta(C)) \geq 1$.
If in addition $x^*(\delta(C)) \geq 3$ holds, then we get $\nicefrac{q(\delta(C))}{2} \geq 1$ regardless of our choice of the multigraph $P$.
If, however, $x^*(\delta(C)) < 3$ holds, we cannot use the same argument.
In that case we need to take care of the constraints of $\Pj{Q_P}$ that correspond to $s$-$t$-cuts, where the $x^*$-load is strictly less than $3$.
Let us denote these cuts by~$\mathcal{B}(x^*)$, that is,
\begin{equation}
\label{eq:Bx}
  \B(x^*) := \left\{ C \subseteq V \midbar s \in C, t \notin C, x^*(\delta(C)) < 3 \right\} \enspace .
\end{equation}

For a family $\B \subseteq \left\{ C \subseteq V \midbar s \in C, t \notin C \right\}$ of $s$-$t$-cuts, we say that a point $y \in \PHK$ is \emph{$\B$-good}, if for every $B \in \B$ we have %either $y(\delta(B)) \geq 3$ \label[type1]{type1}, or $y(\delta(B)) = 1$, and $y$ is integral on the edges $\delta(B)$ \label[type2]{type2}.
\begin{enumerate}[label=(\roman*)]
  \item either $y(\delta(B)) \geq 3$, \label[type]{type1}
  \item or $y(\delta(B)) = 1$, and $y$ is integral on the edges $\delta(B)$. \label[type]{type2}
\end{enumerate}
Therefore, if $y \in \PHK$ is $\B(x^*)$-good, then $q = \nicefrac{x^*}{2} + \nicefrac{y}{2}$ satisfies $\nicefrac{q(\delta(C))}{2} \geq 1$ for every $Q_P$-cut $C$.
We will refer to a cut $B$ satisfying condition (i) as a \textit{\cref{type1} cut}, and if it satisfies condition~(ii) we will refer to it as a \textit{\cref{type2} cut}.
Note that condition~(ii) translates to having a single edge $f \in \delta(B)$ with $y(f)=1$ and $y(e)=0$ for all other edges $e$ from $\delta(B)$.
The notion of $\B$-goodness was introduced by Zenklusen for the elements of the polytope $P_\text{HK}$ in relation to metric {\sc Path TSP}.

\begin{lemma}
\label{thm:characteristicgood}
  The characteristic vector $\chi_U$ of any many-visits $s$-$t$ path $U$ is $\B$-good for any family~$\B$ of $s$-$t$-cuts.
\end{lemma}

\iflong
\begin{proof}
  The lemma easily follows from the fact that a many-visits $s$-$t$ path $U$ crosses any $s$-$t$-cut an odd number of times.
\end{proof}
\fi

We present our algorithm for the metric {\sc Many-visits Path TSP} as \Cref{alg:mvsttsp15}.
\begin{algorithm}[h!]
  \caption{A $\nicefrac32$-approximation algorithm for the metric {\sc Many-visits Path TSP}\label{alg:mvsttsp15}}
  \begin{algorithmic}[1]
    \Statex \textbf{Input:} A complete undirected graph $G=(V,E)$, costs $c:E\rightarrow\mathbbm{R}_{\geq 0}$ satisfying the triangle inequality, requests $r:V\rightarrow\mathbbm{Z}_{\geq 1}$, distinct vertices $s,t\in V$.
    \Statex \textbf{Output:} An $s$-$t$-path that visits each $v \in V$ exactly $r(v)$ times.
    \State Calculate an optimal solution $x^*$ to the Held-Karp relaxation of the {\sc Many-visits Path TSP}, i.e. $x^* := \argmin\{ c^\T x\mid x \in \PHK\}$. \label{st:xstar}
    \State Determine a $\B(x^*)$-good solution $y \in \PHK$ minimizing $c^\T y$. \label{st:Bgood}
    \State Let $B_1\subset\dots\subset B_k$ denote the \cref{type2} cuts with respect to $y$.
    \State Compute a connected multigraph $P$ on $(V, \supp(y))$ such that \label{st:polym_alg}
    \begin{algsubstates}
      \State each vertex $v\in V-\{s,t\}$ has degree at least $2 \cdot r(v)-1$, \label{st:polym_alg_1}
      \State each vertex $v\in\{s,t\}$ has degree at least $2 \cdot r(v)-2$, and \label{st:polym_alg_2}
      \State $P$ contains no parallel edges leaving $B_i$ for $i=1,\dots,k$. \label{st:polym_alg_3}
    \end{algsubstates}
    \State Compute a minimum-cost matching $M$ with respect to $c$ on the vertices $\odd(P) \triangle \{s,t\}$. \label{st:match}
    \State Let $P'$ denote the many-visits path $P+M$. % \label{•}
    \State Do shortcuts in $P'$ and obtain an $s$-$t$-walk $P''$ that visits each city $v$ exactly $r(v)$ times. \label{st:goodtour}
    \Statex \textbf{return} $P''$.
  \end{algorithmic}
\end{algorithm} 

In Step~\ref{st:polym_alg} of the algorithm, we use \Cref{thm:bdmultigraph} to obtain a multigraph with additional properties besides the degree requirements.
In the single-visit counterpart of the problem, one can show that even though $x^*(\delta(B)) < 3$ and $y(\delta(B))=1$ for \cref{type2} cuts $B$, the corresponding point $\nicefrac{q}{2} = \nicefrac{x^*}{4} + \nicefrac{y}{4}$ is still in $\Pj{Q_P}$.
However, due to the possible parallel edges in $P$, the parity argument given by Zenklusen~\cite{Zenklusen2019} does not hold, therefore we need to treat this case separately.
For this reason we make the following distinction.
Let $E_y$ denote the set of edges that correspond to \cref{type2} cuts in $y$, that is
  \begin{equation}
  E_y := \left\{e \in E \mid \exists B \in \mathcal{B}: \supp(y) \cap \delta(B) = e \right\} \enspace .
  \end{equation}
We let $U^\star(e):=1$ for all $e \in E_y$, $U^\star(e):=+\infty$ for the rest of the edges of $\supp(y)$, and $U^\star(e):=0$ for edges $e \in E - \supp(y)$.
Finally, we set $L^\star(e):=0$ for every edge $e \in E$.
According to the claim of \Cref{thm:bdmultigraph}, we can compute a multigraph $P$ satisfying the conditions in \Crefrange{st:polym_alg_1}{st:polym_alg_3}, such that the cost of $P$ is at most $\min\{c^\T x \mid x \in \PSG(\rho^\star, L^\star, U^\star)\}$, where the polytope $\PSG(\rho^\star, L^\star, U^\star)$ depends on the instance $(G,c,r,s,t)$ and can be written in the following form:
  \begin{equation}
  \label{eq:pcg_mvpathtsp}
    \PSG(\rho^\star, L^\star, U^\star) := \left\{ x \in \mathbbm{R}^E_{\geq 0} \midbar \begin{array}{lr}
    \supp(x) \text{ is connected} \\
    x(E) = r(V) \\
    x(\ddelta(v)) \geq 2 \cdot r(v) & \forall v \in V - \{s,t\} \\
    x(\ddelta(v)) \geq 2 \cdot r(v) - 1 & \forall v \in \{s,t\} \\
    0\leq x(e) \leq 1 & \forall e \in E_y \\
    0\leq x(e) \leq +\infty & \forall e \in \supp(y) - E_y \\
    x(e) = 0 & \forall e \in E - \supp(y)
    \end{array} \right\}
  \end{equation}

It is not difficult to see that $y \in \PSG(\rho, L, U)$, and thus

\begin{equation}
  c(P) \leq \min\left\{c^\T x \midbar x \in \PSG(\rho^\star, L^\star, U^\star)\right\} \leq c^\T y  \enspace ,
\end{equation}
therefore $c(P) \leq c^\T y$ holds; this is one of the reasons behind restricting $P$ to $\supp(y)$.
Moreover, according to \Cref{thm:characteristicgood}, the inequality $c^\T y \leq c(\P^\star)$ holds, hence the bound $c(P) \leq c(\P^\star)$ follows.
Now we have all the ingredients to prove our main theorem.

\thmrestate{thm:path15}{\thmpath}{theorem}

\newcommand{\proofthmpath}{
\begin{proof}%[Proof of \Cref{thm:path15}]
  Recall that $Q_P = \odd(P) \triangle \{s,t\}$.
  First prove that $q = \nicefrac{x^*}{2} + \nicefrac{y}{2}$ implies that $\nicefrac{q}{2}$ is in~$\Pj{Q_P}$.
  For that we need to show that  we calculated the solution $y$ in a way that $q$ satisfies $\nicefrac{q(\delta(C))}{2} \geq 1$ for all cuts $C \subset V$ for which $|C \cap \odd(P) \triangle \{s,t\}|$ is odd.
  
  Clearly, $q\in\PHK$, as $q$ is the midpoint of two points from $\PHK$.
  Therefore, for any $Q_P$-cut $C \subseteq V$ that is a not an $s$-$t$-cut, we have $\nicefrac{q(\delta(C))}{2} \geq 1$ as needed.
  Moreover, by definition, for any $Q_P$-cut $C \subseteq V$ that is an $s$-$t$-cut and is not included in $\B(x^*)$, we have $x^*(\delta(C)) \geq 3$, and so
  \begin{equation}
  \frac12 q(\delta(C)) = \frac14 \big(x^*(\delta(C)) + y(\delta(C))\big) \geq 1 ,
  \end{equation}
  as $y \in \PHK$, and thus $y(\delta(C)) \geq 1$.
  
  It remains to consider $Q_P$-cuts $C \subseteq V$ that are in $\B(x^*)$.
  Since $y$ is $\B(x^*)$-good by construction, either $y(\delta(C)) \geq 3$, or  $y(\delta(C)) = 1$ with $y$ being integral on the edges $\delta(C)$.
  If $y(\delta(C)) \geq 3$, then $\nicefrac{q(\delta(C))}{2} \geq 1$ follows from $x^*(\delta(C)) \geq 1$ and the definition of~$q$.
  If $y(\delta(C)) = 1$ and $y$ is integral on the edges $\delta(C)$, it holds that $y(e)=0$ for all edges of $\delta(C)$ except for one $f \in \delta(C)$ where $y(f)=1$.
  It is at this point where we exploit the restrictions imposed on $P$.
  Since $\supp(P) \subseteq \supp(y)$, and the load on an edge $e \in E_y$ is at most $1$ in $P$, the only edge of $P$ with a positive load on $\delta(C)$ is~$f$, and that load is at most $1$.
  Moreover, every cut has at least 1 load in $P$, which means $|P \cap \delta(C)| = 1$.
  But an $s$-$t$-cut $C \subseteq V$ with $|\delta_P(C)|$ odd cannot be a $Q_P$-cut because of the following:
  \begin{equation}
  \label{eq:parity}
    \begin{aligned}
      |C \cap \odd(P)| & \equiv \sum_{v \in C} |\dot\delta_P(v)| \pmod{2} \\
                       & = 2 \cdot \left|\left\{ uv \in P\mid u, v \in C\right\}\right| + |\delta_P(C)| \enspace .
    \end{aligned}
  \end{equation}

  \Cref{eq:parity} implies that $|C \cap \odd(P)|$ is odd, and hence $|C \cap Q_P| = |C \cap (\odd(P) \triangle \{s,t\})|$ is even because~$C$ is an $s$-$t$-cut.
  By the above, any cut of \cref{type2} partitions the vertices of $\odd(P) \triangle \{s,t\}$ into two subsets of even cardinality.\footnote{For a cut $C$ with $y(\delta(C)) = 1$ and $y$ being integral on $\delta(C)$, the term $|T \cap \delta(C)|$ in the proof of Theorem~2.1 of Zenklusen~\cite{Zenklusen2019} corresponds to the term $|\delta_P(C)|$ in \Cref{eq:parity}.
  Since the spanning tree $T$ computed on $\supp(y)$ in the algorithm of~\cite{Zenklusen2019} cannot contain parallel edges, $|T \cap \delta(C)|$ has a value of 1 without enforcing an upper bound on the edge $e \in \delta_T(C)$ for $y(e)=1$.}
  This means that no cut constraint of~$\Pj{Q_P}$ requires a load of $1$ for $\nicefrac{q}{2}$ on $C$, and so $\nicefrac{q}{2} \in \Pj{Q_P}$ holds.   

  The cost of the matching $M$ can therefore be bounded as follows:
  \begin{equation}
    c(M) \leq c\left(\frac{q}{2}\right) = \frac14 c^\T x^* + \frac14 c^\T y \leq \frac12 c(\P^\star),
  \end{equation}
  since $c^\T x^* \leq c(\P^\star)$.
  Thus, the multigraph obtained from $P+M$ has cost at most $\nicefrac{3}{2} \cdot c(\P^\star)$, as claimed.
  
  \paragraph{Shortcuts and complexity.}
  According to \Cref{thm:bdmultigraph}, every vertex $v \in V - \{s,t\}$ has degree at least $2\cdot r(v)-1$, while vertices $s$ and $t$ have degrees at least $2 \cdot r(s)-2$ and $2 \cdot r(t)-2$ respectively, in the multigraph $P$.
  The matching $M$ provides 1 additional degree for vertices with the wrong parity, therefore $P'$ will have an even degree at least $r(v)$ for all $v \in V - \{s,t\}$ and an odd degree at least $r(v)-1$ for $v \in \{s,t\}$.
  This means that $P'$ corresponds to a many-visits $s$-$t$-path that visits each vertex $v$ at least $r(v)$ times, but possibly more.
  In \Cref{st:goodtour} we proceed with taking shortcuts the way described in \Cref{alg:apx_tp_path}, so that $P''$ is a feasible solution to the {\sc Many-visits Path TSP} instance $(G,c,r,s,t)$.

  Now we turn to the complexity analysis.
  The constraints of the Held-Karp relaxation (\Cref{eqn:mvpathtsp_hkrelax}) of the {\sc Many-visits Path TSP} can be tested in time polynomial in $n$ and $\log r(V)$, hence calculating~$x^*$ takes a $\poly(n, \log r(V))$ time as well~\cite[\S 58.5]{Schrijver2003}.
  This means \Cref{st:xstar} takes time polynomial in~$n$ and $\log r(V)$.
  According to \Cref{lem:bgood}, \Cref{st:Bgood} also has polynomial time complexity, and 
  By \Cref{thm:bdmultigraph}, \Cref{st:polym_alg} takes polynomial time and calculating a matching in \Cref{st:match} can be done efficiently as well.
  Finally, since the number of edges in $P$ is $r(V)$ and the matching $M$ contributes at most $\nicefrac{n}{2}$ edges, we remove at most $\nicefrac{n}{2}$ edges from $P'$ to obtain our solution $P''$.
  This means that the number of operations performed in \Cref{st:goodtour} can be bounded by $\O(n)$.
  The claimed time complexity follows.
\end{proof}
}

\iflong \proofthmpath \else \seeinappendix \fi

\thmrestate{thm:tsp15}{\thmtsp}{corollary}

\newcommand{\proofthmtsp}{
\begin{proof}
  Let $G=(V, E)$ be a graph and $(G, c, r)$ denote a metric {\sc Many-visits TSP} instance.
  Choose an arbitrary vertex $v \in V$, and construct a metric {\sc Many-visits Path TSP} instance $(\hat{G}, \hat{c}, \hat{r}, s_v, t_v)$ as follows.
  Let $\hat{G}$ be an undirected graph on the vertex set $\hat{V} := V - v \cup \{s_v, t_v\}$ and edge set $\hat{E} := \hat{V} \times \hat{V}$.
  We define $\hat{c}(s_v t_v)$ as the cost of a self loop at $v$, $c(vv)$, and $\hat{c}(s_v u) = \hat{c}(t_v u) := c(vu)$ for every vertex $u \in V - v$.
  Moreover, the self-loops at $v_s$ and $v_t$ have cost $c(vv)$ as well.
  It is easy to check that $\hat{c}$ satisfies the triangle inequality.
  Finally, set $\hat{r}(s_v) := r(v)$ and $\hat{r}(t_v) := 1$.

  Now we prove that the {\sc Many-visits TSP} instance $(G, c, r)$ and the corresponding {\sc Many-visits Path TSP} instance $(\hat{G}, \hat{c}, \hat{r}, s_v, t_v)$ can be reduced to each other.
  First let $T$ be a solution to $(G, c, r)$.
  Choose an edge $vu$ from $T$ such that $v\neq u$, and let $P$ denote the many-visits $s_v$-$t_v$-path obtained from $T$ by deleting $v$, replacing each occurrence of any edge $vw\in T$ with a copy of $s_vw$ if $w\in V-\{u,v\}$ or with a copy of the loop on $s_v$ if $w=v$, and replacing all but one occurrence of the edge $vu\in T$ with a copy of $s_vu$, while one copy of $vu$ is substituted by $t_vu$.
  In other words, $s_v$ \enquote*{inherits} all copies of all edges and self-loops incident to $v$, except one copy of~$uv$, and $t_v$ inherits one copy of $uv$.
  This means that $\deg_P(s_v) = 2 \cdot r(v)-1$ and $\deg_P(t_v) = 1$.
  Note that each edge of $T$ is replaced by an edge of the same cost, and every vertex $w \in V - v$ has the same degree in $T$ and $P$, hence $\deg_P(w) = 2 \cdot r(w)$.
  Therefore, $P$ is a feasible solution to $(\hat{G}, \hat{c}, \hat{r}, s_v, t_v)$ of the same cost as $T$.

  Now consider a multigraph $P$ that is a solution to $(\hat{G}, \hat{c}, \hat{r}, s_v, t_v)$.
  Identify the vertices~$s_v$ and $t_v$, denote the new vertex by $v$, and introduce an edge $vv$ for every copy of the edge~$s_v t_v$ in $P$.
  Let us denote the resulting multigraph by $T$.
  Since $c(u v) = \hat{c}(s_v u) = \hat{c}(t_v u)$ for all $u \in V - v$ and $c(vv) = \hat{c}(s_v s_v) = \hat{c}(t_v t_v) = \hat{c}(s_v t_v)$, replacing $s_v$ and $t_v$ by $v$ the way described above does not change the cost of the multigraph.
  Moreover, the degree of $v$ in~$T$ is $\deg_T(v) = \deg_P(s_v) + \deg_P(t_v) = 2 \cdot r(v) - 1 + 1 = 2 \cdot r(v)$.
  The degrees of vertices $w\in V - v$ remain unchanged, thus $T$ is a feasible solution to $(G, c, r)$ of the same cost as $P$.

  We therefore showed that for every solution of $(G, c, r)$ there exists a solution of $(\hat{G}, \hat{c}, \hat{r}, s_v, t_v)$ with the same cost, and vice versa.
  Let now $(G, c, r)$ be a metric {\sc Many-visits TSP} instance.
  Pick an arbitrary vertex $v\in V$, and consider the corresponding metric {\sc Many-visits Path TSP} instance $(\hat{G}, \hat{c}, \hat{r}, s_v, t_v)$, and obtain a $\nicefrac32$-approximation $P$ using \Cref{alg:mvsttsp15}.
  Identify $s_v$ and~$t_v$ into $v$ again, and substitute each copy of the edge $s_v t_v$ in $P$ by a copy of the self-loop $vv$.
  By the above, the resulting multigraph $T$ gives a $\nicefrac32$-approximation to the instance $(G, c, r)$.
\end{proof}
}

\newcommand{\remarkmvtsp}{
\begin{remark}
  Alternatively, one can directly obtain a $\nicefrac32$-approximation for the metric {\sc Many-visits TSP} by performing \Cref{st:polym_alg_1}, \Cref{st:match} and \Cref{st:goodtour} of \Cref{alg:mvsttsp15}.
  More precisely, calculate a connected multigraph $T$ with degrees at least $2 \cdot r(v) -1$ and cost at most the optimum using the result of \Cref{thm:bdmultigraph}, then calculate a matching on the odd degree vertices and apply shortcuts.
  This procedure was described by a superset of the authors~\cite{BercziBMV2019}.
\end{remark}
}

\iflong \proofthmtsp \remarkmvtsp \else \seeinappendix \fi

Before we show how to calculate a $\B(x^*)$-good point $y \in \PHK$, let us show that the number of cuts in $\B$ is polynomial in $n$, and that the set $\B$ can be computed efficiently:

\newcommand{\lemncutsb}{
  Let $q \in \PHK$.
  Then the family $\B(q)$ of $s$-$t$-cuts of $q$-value strictly less than 3 satisfies $|\B(q)| \leq n^4$ and can be computed in $\O(mn^4)$ time, where $n := |V|$ and $m := \supp(q)$.
}

\begin{lemma}
\label{lem:n_cuts_B}
  \lemncutsb
\end{lemma}

\newcommand{\prooflemncutsb}{
\begin{proof}
  Let us define an auxiliary graph $H=(V, E')$ whose edge set~$E'$ consists of the edges in $\supp(q)$ and an additional $st$ edge.
  Let $q_H=q+\chi_{st}$.
  Clearly, for non-$s$-$t$-cuts we have $q_H(\delta_H(C)) = q(\delta(C))$, while for $s$-$t$-cuts we have $q_H(\delta_H(C)) = q(\delta(C)) + 1 \geq 2$ because of the newly added edge $st$. 
  Therefore, the family $\B(q)$ can be written as
  \begin{equation*}
    \B(q) = \{C \subset V \mid s \in C, t \notin C, q_H(\delta_H(C)) < 4 \} \enspace .
  \end{equation*}
  The minimum cut has a load of at least $2$, and due to Karger~\cite{Karger1993} the number of cuts with a load less than $k$ times the minimum cut is at most $\O(n^{2k})$.
  Moreover, using an algorithm by Nagamochi et al.~\cite{NagamochiNI1997}, we can enumerate the cuts of size at most $k$ times the minimum cut in time $\O(m^2 n + n^{2k} m)$.
  These results prove that the number of cuts in $|B(q)|$ is $\O(n^4)$, and that they can be enumerated in time $\O(m n^4)$.
\end{proof}
}

\iflong \prooflemncutsb \else \seeinappendix \fi

\subsubsection*{The dynamic program}

Given a family $\B$ of $s$-$t$-cuts, our goal is to determine a minimum cost $\B$-good point $y \in \PHK$.
We use the dynamic programming approach introduced by Traub and Vygen~\cite{TraubVygen2020} and improved upon by Zenklusen~\cite{Zenklusen2019}.
More precisely, the goal of the dynamic program is to determine which cuts in $\B$ are of \cref{type1}, and which ones are of \cref{type2}.
Our approach is constructive as the dynamic program also determines a point $y$ that is $\B(x^*)$-good.

Consider a $\B(x^*)$-good point $y$.
Let $B_1,\dots,B_k$ denote the \cref{type2} $s$-$t$ cuts in $\B$ with respect to~$y$, that is, $y(v_i u_i)=1$ for exactly one edge $v_i u_i \in\delta(B_i)$ and $y(e)=0$ for $e\in\delta(B_i)-v_i u_i$.
It is not difficult to see that these cuts necessarily form a chain (see e.g.~\cite{Zenklusen2019}), thus we set the indices such that $B_1\subsetneq\dots\subsetneq B_k$. The endpoints of $v_i u_i$ are named such that $v_i\in B_i$, $u_i\notin B_i$. Furthermore, we define $B_0 := \emptyset$, $B_{k+1}:= V$, $u_0:=s$ and $v_{k+1}:=t$ for notational convenience. Note that $u_i$ and $v_{i+1}$ might coincide for some $i=0,\dots,k+1$.

The work of Zenklusen~\cite{Zenklusen2019} argues that the \enquote*{first} and \enquote*{last} cuts are \cref{type2} cuts, that is, $B_1 = \{s\}$ and $B_k = V - \{t\}$, because the constraints of $\PHK$ enforce a degree of $1$ on vertices $s$ and~$t$.
In the many-visits setting, however, this is not necessarily true, as the instance possibly requires more than one visit for $s$ or $t$.

Assume for a moment that we knew the cuts $B_1, \dots, B_k$ and the edges $v_i u_i$, and we are looking for a $\B(x^*)$-good point $y \in \PHK$ such that among all cuts in $\B$ the cuts $B_1, \dots, B_k$ are precisely those where (a) $y$ is integral, and (b) $y(\delta(B_i))=1$ for all $i=1, \dots, k$.
Then the $\B$-good points $y \in \PHK$ that satisfy these constraints (a) and (b) have the following properties for all $i=1, \dots, k$:
\begin{enumerate}[label=(P\arabic*)]
  \item $y(v_i u_i)=1$ and $y(e)=0$ for all edges $e \in \delta(B_i) - v_i u_i$,  \label{pro_bgood_1}
  \item the restriction of $y$ to the vertex set $B_{i+1} - B_i$ is a solution to the Held-Karp relaxation for the {\sc Many-visits Path TSP} with endpoints $u_i$ and $v_{i+1}$, with the additional property that $y(\delta(B)) \geq 3$ for every cut $B \in \B$ such that $B_i \cup u_i \subseteq B \subseteq B_{i+1} - v_{i+1}$. \label{pro_bgood_2}
\end{enumerate}

The dynamic program thus aims to find cuts $B_1, \dots, B_k$ while exploiting the properties \ref{pro_bgood_1} and \ref{pro_bgood_2} above.
Formally, it is defined to find a shortest path on an auxiliary directed graph.
Let us define the auxiliary directed graph $H=(N,A)$ with node set $N$, arc set $A$, and length function $d: A \rightarrow \mathbbm{R}_{\geq 0}$.
The node set $N$ is defined by $N=N^+ \cup N^-$, where 
\begin{align*} 
 N^+ {}&{}= \left\{(B,u) \in \B \times V \midbar u \notin B \right\} \cup \left\{(\emptyset,s)\right\}\text{, and} \\
 N^- {}&{}= \left\{(B,v) \in \B \times V \midbar v \in B \right\} \cup \left\{(V,t)\right\} \text{.}
\end{align*}
The arc set $A$ is given by $A=A_\text{HK} \cup A_E$, where
\begin{align*}
  A_\text{HK} {}&{}= \left\{ \left((B^+,u), (B^-,v)\right) \in N^+ \times N^- \midbar B^+ \subseteq B^-, u,v \in B^- - B^+ \right\} \text{, and} \\
 A_E {}&{}= \left\{ \left( (B^-, v), (B^+, u)\right) \in N^- \times N^+ \midbar B^- = B^+ \right\} \text{.}
\end{align*}
Finally, the lengths $d: A \rightarrow \mathbbm{R}_{\geq 0}$ are defined as follows:
\begin{equation*}
d(a)=
\begin{cases}
c(vu) & \text{if}\ a = \left( (B,v), (B,u) \right) \in A_E,\\
\OPT(\LP(a)) & \text{if}\ a \in A_\text{HK},
\end{cases}
\end{equation*}
where $\OPT(\LP(a))$ denotes the optimum value of 
\begin{align}
  \min\quad        & c^\T x                                                   \nonumber       \\
  \text{subject to \quad} & x \in \PHK(B^- - B^+, u, v) &   &  \tag{$\LP(a)$} \label[empty]{eq:LPa}\\
                   & x(\delta(B)) \geq 3                               &   & \text{for all}\ B \in \B\ \text{s.t.}\ B^+ \subseteq B \subseteq B^-,\nonumber\\
                   & & &  u\in B, v \notin B \text{,}\nonumber
%  \min c^\T x \quad \text{s.t. \quad} & x \in \PHK(B^- - B^+, u, v) &   &  \tag{$\LP(a)$} \label[empty]{eq:LPa}\\
%                   & x(\delta(B)) \geq 3                                  & \text{for all}\ B \in \B\ \text{s.t.}\ B^+ \subseteq B \subseteq B^-, \nonumber \\ & &  u \in B, v \notin B \text{,}\nonumber
\end{align}
where $a=\left((B^+,u), (B^-,v)\right)$.

For $y$-values across the cuts $B \in \B$ so that $B \notin \{B_1, \dots, B_k\}$, we require that $y(\delta(B)) \geq 3$ holds.
We ensure this by defining modified Held-Karp relaxations of the {\sc Many-visits Path TSP} instances between cuts $B_i$ and $B_{i+1}$ for every $i=0, \dots, k$.
More precisely, such an instance is defined on the subgraph of $G$ induced on the vertex set $B_{i+1} - B_i$ with distinguished vertices~$u_i$ and $v_{i+1}$, with the additional property that it has a $y$-load of at least 3 on each cut $B \in \B$ with $B_i \subset B \subset B_{i+1}$, as shown in \Cref{eq:LPa}.
In case $u \neq v$, the polytope $\PHK(W,u,v)$ is defined as follows:

\begin{equation}
\label{eq:phkmv_different}
  \PHK(W,u,v) := \left\{ x \in \mathbbm{R}^E_{\geq 0} \midbar \begin{array}{ll}
                       x(\delta(C)) \geq 2  & \forall C \subset W, C \neq \emptyset, \\ 
 & \quad |C \cap \{u, v\}| \in \{0, 2\} \\
x(\delta(C)) \geq 1  & \forall C \subset W, |C \cap \{u, v\}| = 1 \\
x(\ddelta(w)) = 2 \cdot r(w) & \forall w \in W - \{u, v\}\\
x(\ddelta(u)) = 2 \cdot r(u) - 1 \\
x(\ddelta(v)) = 2 \cdot r(v) - 1 \\
x(e) = 0 & \forall e \in E - E[W]
\end{array} \right\}
\end{equation}

Note that unlike in the single-visit case~\cite{Zenklusen2019}, we allow $u$ being equal to $s$ or $v$ being equal to~$t$ in \Cref{eq:phkmv_different}, and the corresponding polytopes $\PHK(B_1, s, v_1)$ and $\PHK(V - B_k, u_k, t)$ are feasible.

Let us now cover the case when for some index $i \in \{0, \dots, k\}$, vertices $u_i$ and $v_{i+1}$ coincide.~\footnote{
Note that the corresponding arc in $H$ will have the form $\left( (B^+, w), (B^-, w) \right) \in A_\text{HK}$.
}
In the single-visit {\sc Path TSP}, the solution is defined to be the all-zero vector if $u_i=v_{i+1}$ is the only vertex in $B_{i+1} - B_i$, and there exists no solution otherwise.
However, since we allow for a vertex to be visited more than once (i.e. have a degree more than~$2$) in a solution to the Held-Karp relaxation for the {\sc Many-visits Path TSP}, we use a different extension in our approach.
We define the corresponding subproblem as the Held-Karp relaxation for the {\sc Many-visits TSP}.
First assume that $u_i \neq s$ and $u_i \neq t$.
Since $y(v_i u_i)=1$ and $y(u_i u_{i+1})=1$ by construction, the degree requirement for $u_i$ in the {\sc Many-visits TSP} subproblem is \emph{two} less that in $\PHK$, namely $r(u_i)-2$.
If $u_i = s$ (or $u_i = t$), then due to $y(s \, u_1)=1$ (or $y(v_k \, t)=1$) the degree requirement for $u_i$ in the subproblem is \emph{one} less that in $\PHK$, which also equals to $r(u_i)-2$.
Note that if $u_i = v_{i+1}$ there is no cut $B \in \B$ with $u_i \in B$ and $ v_{i+1} \notin B$, thus the linear program $\LP(a)$ has the form $\min \{ c^\T x \mid x \in \PHK(W,u,u) \}$, where:

\begin{equation}
\label{eq:phkmv_same}
  \PHK(W,u,u) := \left\{ x \in \mathbbm{R}^E_{\geq 0} \midbar \begin{array}{ll}
                       x(\delta(C)) \geq 2  & \forall C \subset W, C \neq \emptyset, \\ 
% & \quad |C \cap \{u, v\}| \in \{0, 2\} \\
%x(\delta(C)) \geq 1  & \forall C \subset W, |C \cap \{u, v\}| = 1 \\
x(\ddelta(w)) = 2 \cdot r(w) & \forall w \in W - u\\
x(\ddelta(u)) = 2 \cdot r(u) - 2 \\
x(e) = 0 & \forall e \in E - E[W]
\end{array} \right\} \enspace .
\end{equation}

If the requirement for $u$ is $r(u)=1$ and $|W|>1$, the polytope $\PHK(W,u,u)$ is empty, and thus the linear program $\LP(a)$ has no solution. 
In this case the cost of the arc $a$ is defined to be infinity.
Note however that if $r(u)=1$ and $W = \{u\}$, the corresponding linear program has a non-zero solution, namely a vector that has value $r(u)-1$ in the coordinate of the self-loop $uu$, and $0$ otherwise.

To find a $\B$-good point with minimum cost $c^\T y$, we compute a shortest $(\emptyset, s)$--$(V,t)$ path with respect to $d$ in $H$; due to \Cref{thm:characteristicgood,lem:lstar} this path has finite length. Let $(\emptyset, s), (B_1, v_1), \allowbreak (B_1, u_1), (B_2, v_2), \dots, (B_k, u_k), (V, t)$ be the nodes on this shortest path, and similarly as before, define $B_0 := \emptyset, u_0 := s$ and $B_{k+1} := V, v_{k+1} := t$.
By construction of $H$, we have $B_0 \subset B_1 \subset \dots \subset B_{k+1}$. 
Let $x^i \in \mathbbm{R}^E$ be an optimal solution to $\LP(a)$ for $a = ((B_i,u_i),(B_{i+1},v_{i+1}))$.
Set 
\begin{equation}
\label{eqn:thisway}
  y := \sum_{i=0}^k x^i + \sum_{i=1}^k \chi_{v_i u_i} .
\end{equation}
By the definition of the lengths $d$ in $H$, $c^\T y$ necessarily equals the length $\ell^*$ of a shortest $(\emptyset, s)$--$(V,t)$ path in $H$ with respect to $d$.
We now show that $y$ computed in \Cref{eqn:thisway} is indeed a $\B(x^*)$-good point of minimum cost.

\newcommand{\lemlstar}{
  The length $\ell^*$ of a shortest $(\emptyset, s)$--$(V,t)$ path in $H$ with respect to $d$ satisfies $\ell^* \leq \min\{c^\T z \mid  z \in \PHK, z \text{ is } \B \text{-good}\}$.
}

\begin{lemma}
\label{lem:lstar}
  \lemlstar
\end{lemma}

\newcommand{\prooflemlstar}{
\begin{proof}
  Let $\B_z \subseteq \B$ be the family of cuts $B \in \B$ such that $z(f)=1$ for precisely one edge $f \in \delta(B)$, and $z(e)=0$ for all other edges $e \in \delta(B) - f$.
  These are the sets in $\B$ that are \cref{type2} cuts with respect to $z$, and also $\B_z$ forms a chain: $B_1 \subset \dots \subset B_k$ holds, where $B_i \in \B_z$ for $i=1,\dots,k$.
  The cuts $\{B_1, \dots, B_k\}$ defines a partition of $V$ into sets $B'_0 := B_1$, $B'_1 := B_2 - B_1$, $\dots$, $B'_{k-1} := B_k - B_{k-1}$, $B'_k := V - B_k$.
  For $i \in \{1, \dots, k\}$, let $v_i u_i$ be the unique edge in $\delta(B_i)$ where $z(v_i u_i)=1$, so that $v_i \in B_i$ and $u_i \notin B_i$.

  Consider the path along nodes $(B_0, u_0)$, $(B_1, v_1)$, $(B_1, u_1), \dots, (B_{k+1}, v_{k+1})$.
  It suffices to show that the length $\ell$ of the path is at most $c^\T z$.
  For each $i \in \{0, \dots, k\}$, the vector $z^i \in \mathbbm{R}^E$ is defined to be the restriction of $z$ to $E[B_{i+1} - B_i]$.
  Assume for a moment that $z^i$ is a feasible solution of $\LP(a)$ with $a=((B_i, u_i), (B_{i+1},v_{i+1}))$.
  Then the total length $\ell$ is equal to $\sum_{i=0}^k c^\T x^i + \sum_{i=1}^k c(v_i u_i)$ by definition, which is at most $\sum_{i=0}^k c^\T z^i + \sum_{i=1}^k c(v_i u_i) = c^\T z$.
  Since~$\ell^*$ is minimum among all possible $\ell$'s, we get $\ell^* \leq \ell \leq c^\T z$.

  Since $z$ is $\B$-good, and $z^i(\delta(B)) = z(\delta(B))$ for any cut $B$ with $B_i \subsetneq B \subsetneq B_{i+1}$, that means $z^i(\delta(B)) = z(\delta(B)) \geq 3$.
  It remains to show that $z^i \in \PHK(B'_i,u_i,v_{i+1})$ or $z^i \in \PHK(B'_i,u_i,u_i)$ follows for $i=0,\dots,k$.
    
\paragraph{Distinct endpoints.} 
  Let us start with the case when $u_i \neq v_{i+1}$.
  By definition, $z^i(e) = z(e)$ if both endpoints of $e$ are in $B'_i$ and $z^i(e) = 0$ otherwise.
  Hence, for vertices $w \in B'_i$ such that $w \notin \{u_i, v_{i+1}\}$, $z^i(\ddelta(w)) = z(\ddelta(w)) = 2 \cdot r(w)$.
  First assume that $u_i \neq s$ and $v_{i+1} \neq t$.
  Both of the endpoints $u_i$ and $v_{i+1}$ have a total $z$-value of $1$ on the edge set $\delta(B'_i)$ due to $z(v_i u_i)=1$ and $z(v_{i+1} u_{i+1})=1$, therefore the $z^i$-value on the edges incident to the endpoints is $z^i(\ddelta(u_i)) = z(\ddelta(u_i))-1 = 2 \cdot r(u_i) -1$ and $z^i(\ddelta(v_{i+1})) = z(\ddelta(v_{i+1}))-1 = 2 \cdot r(v_{i+1}) -1$.
  If $u_i = s$ or $v_{i+1} = t$, their $z^i$-values equal to their $z$-values, that is $2 \cdot r(s) - 1$ or $2 \cdot r(t) - 1$, respectively.
  The degree constraints are therefore satisfied.  

  Finally, we have to show that for a cut $C \subseteq B'_i$, $z^i(\delta(C)) \geq 1$ holds if $C$ is a $u_i$-$v_{i+1}$-cut, and $z^i(\delta(C)) \geq 2$ if $C$ does not separate $u_i$ and $v_{i+1}$.
  For the single-visit variant, the proof goes by showing that $z^i$ is in the spanning tree polytope of $G[B'_i]$, using the fact that $z$ is in the spanning tree polytope of $G$ and the degree constraints of $v \in B'_i$~\cite{Zenklusen2019}.
  However, these terms do not immediately generalize to the many-visits setting, so we show that the connectivity of $z^i$ follows from the properties of $z$.
  
  \textbf{Non-$u_i$-$v_{i+1}$-cuts:}  
  First let us consider the cuts that do not separate $u_i$ and $v_{i+1}$, and prove that the total $z^i$-value across these cuts is at least~$2$.
  We may assume that $C \subset B'_i$ does not contain neither $u_i$ nor $v_{i+1}$ throughout this paragraph.
  In case $u_i, v_{i+1} \in C$, we can take $B'_i - C$ and we are done, as $z^i(\delta(C)) = z^i(\delta(B'_i - C)) = z(\delta(C)) \geq 2$, yielding $z^i(\delta(C)) \geq 2$.
  Assume first that $u_i \neq s$ and $v_{i+1} \neq t$, and let $C \subset B'_i$.
  Then $z^i(\delta(C)) \geq 2$ simply because $C$ is a non-$s$-$t$-cut and thus $z(\delta(C)) \geq 2$.
  Now let $u_i$ be equal to $s$, and let $C \subset B'_i$ be a cut that does not contain either $s$ or $v_1$.
  Likewise, $z^0(\delta(C)) \geq 2$ because $C$ is a non-$s$-$t$-cut and thus $z(\delta(C)) \geq 2$.
  The argument for $v_{i+1} = t$ goes the same way.
  
  \textbf{$u_i$-$v_{i+1}$-cuts:}
  Here we have to prove that if $C$ is a $u_i$-$v_{i+1}$-cut, then $z^i(\delta(C)) \geq 1$.

  Assume that $u_i \neq s$ and $v_{i+1} \neq t$.
  Without the loss of generality, we may assume that $u_i \in C$.
  Since $C$ is a non-$s$-$t$-cut, $z(\delta(C)) \geq 2$.
  If we account for $z(v_i u_i)=1$, then $z^i(\delta(C)) \geq 1$ follows.
  One can similarly prove the claim if $u_i \notin C$ and $v_{i+1} \in C$.
  Assume now that $u_i = s$ and $u_i \in C$; then $C$ is an $s$-$t$-cut, so $z(\delta(C)) \geq 1$.
  Moreover, $v_{i+1} \notin C$, so $v_{i+1} u_{i+1} \notin \delta(C)$, therefore $z^i(\delta(C)) = z(\delta(C)) \geq 1$.
  If $u_i = s$ and $v_{i+1} \in C$, then $z^i(\delta(C)) \geq z(\delta(C)) - z(v_{i+1} u_{i+1}) \geq 1$, since $C$ is a non-$s$-$t$-cut and thus $z(\delta(C)) \geq 2$.
  The case $v_{i+1} = t$ can be proved similarly.
  
  \paragraph{Same endpoints.}
  Now we cover the case when $u_i = v_{i+1}$.
  We need to prove that $z^i \in\PHK(B'_i, u_i, u_i)$, as defined in \Cref{eq:phkmv_same}.
  The argument about the degrees is analogous to the case above, the fact that $z^i(\ddelta(w)) = 2 \cdot r(w)$ directly follows for vertices $w \in B'_i - u_i$.
  Moreover, if $u_i \notin \{s,t\}$, the endpoint $u_i$ has a $z$-load of $2$ on $\delta(B'_i)$ because of $z(v_i u_i) = 1$ and $z(v_{i+1} u_{i+1}) = z(u_i u_{i+1}) = 1$, hence $z^i(\ddelta(u_i)) = z(\ddelta(u_i)) - 2 = 2 \cdot r(u_i) -2$, as desired. 
  If $u_i \in \{s,t\}$, then $z(\ddelta(u_i)) = 2 \cdot r(u_i) - 1$.
  Moreover, $z(s \, u_1)=1$ or $z(v_k \, t)=1$ if $u_i=s$ or $u_i=t$, respectively.
  In both cases this means $z^i(\ddelta(u_i)) = z(\ddelta(u_i)) - 1 = 2 \cdot r(u_i) - 2$.

  Now turn to the cut constraints, and let $C \subset B'_i$ be a cut.
We can assume that $u_i \notin C$, otherwise we take $B'_i - C$, and we are done.
  If $u_i \notin C$, then $z^i(\delta(C)) = z(\delta(C)) \geq 2$ because $C$ is a non-$s$-$t$-cut.
  The proof is complete.
\end{proof}
}

\iflong \prooflemlstar \else \seeinappendix \fi

\newcommand{\lempolytopecontainment}{
$y \in \PHK$.
}

\begin{lemma}
\label{lem:polytopecontainment}
  \lempolytopecontainment
\end{lemma}

\newcommand{\prooflempolytopecontainment}{
\begin{proof}
To prove this claim, we use the properties of $x^i$.
We again distinguish two cases, based on whether the endpoints of subproblem $\LP(a)$ are the same or different.
\paragraph{Degree constraints.}
  First, consider the indices $i$ with $u_i \neq v_{i+1}$.  
  By definition, the vector $x^i$ satisfies $x^i \in \PHK(B_{i+1} - B_i, u_i, v_{i+1})$, meaning that it is a solution to the Held-Karp relaxation for {\sc Many-visits Path TSP} in the induced subgraph $G[B_{i+1} - B_i]$ with endpoints~$u_i$ and $v_{i+1}$.
  Recall that $y = \sum_{i=0}^k x^i + \sum_{i=1}^k \chi_{v_i u_i}$ by definition.
  For a $v \in B_{i+1} - B_i$, the value $y(\ddelta(v))$ is equal to 
  \begin{enumerate}
    \item $x^i(\ddelta(v))$ if $v=u_i=s$ or $v=v_{i+1}=t$,
    \item $x^i(\ddelta(v)) + 1 = 2 \cdot r(v)$ if $v \notin \{s,t\}$ and $v = u_i$ or $v = v_{i+1}$ for some $i$, due to the edge $v_i u_i$ or $v_{i+1} u_{i+1}$, respectively; and
    \item $x^i(\ddelta(v)) = 2 \cdot r(v)$ otherwise.
  \end{enumerate}
  By the above, $y(\ddelta(s)) = 2\cdot r(s)-1$, $y(\ddelta(t)) = 2\cdot r(t) -1$, and $y(\ddelta(v)) = 2 \cdot r(t)$ for $v \notin \{s,t\}$, therefore the degree constraints are satisfied for all $v \in V$.

  If $u_i = v_{i+1}$, the value $y(\ddelta(v))$ is equal to
  \begin{enumerate}
    \item $x^i(\ddelta(v)) + 1$ if $v=s$ or $v=t$, because of the edge $v u_1$ or $v_k v$, respectively,
    \item $x^i(\ddelta(v)) + 2 = 2 \cdot r(v)$ if $v \notin \{s,t\}$ and $v = u_i = v_{i+1}$ for some $i$, due to the edge $v_i u_i$ and $v_{i+1} u_{i+1}$; and    
    \item $x^i(\ddelta(v) = 2 \cdot r(v)$ otherwise.
  \end{enumerate}
  Again, we get $y(\ddelta(s)) = 2 \cdot r(s) -1$, $y(\ddelta(t)) = 2\cdot r(t) -1$, and $y(\ddelta(v)) = 2 \cdot r(t)$ for $v \notin \{s,t\}$, therefore the degree constraints are satisfied for all $v \in V$. 

\paragraph{Cut constraints.}

  It remains to show that $y$ satisfies the cut constraints.
  As in the proof of \Cref{lem:lstar}, instead of building on a spanning subgraph polytope, we directly prove that the cut constraints hold.
  As before, the cuts $\{B_1, \dots, B_k\}$ define a partition of $V$ into sets $B'_0 := B_1$, $B'_1 := B_2 - B_1$, $\dots$, $B'_{k-1} := B_k - B_{k-1}$, $B'_k := V - B_k$.
  
  Let us first consider the value of $y$ on $s$-$t$-cuts.
  For $B_i \in \{B_1, \dots, B_k\}$ the $y$-load on $\delta(B_i)$ equals to $1$ due to the edge $v_i u_i$, therefore it satisfies the constraint $y(\delta(B_i)) \geq 1$.
  If $C$ is a $s$-$t$-cut such that $C \notin \{B_1, \dots, B_k\}$, then there is at least one index $i \in \{0, \dots, k\}$, such that  both $B'_i \cap C$ and $B'_i - C$ are not empty.
  In other words, there is at least one vertex from $B'_i$ on both sides of the cut $C$.

  If $u_i \neq v_{i+1}$, $x^i$ satisfies the constraints of $\LP$ for $a = \left((B_i, u_i), (B_{i+1}, v_{i+1}) \right)$, we have $x^i \in \PHK(B'_i, u_i, v_{i+1})$.
  That means $x^i$ has a load of at least $1$ on edges leaving every proper subset of $B'_i$, including $B'_i \cap C$, and $x^i(\delta(B'_i \cap C)) \geq 1$ implies $y(\delta(B'_i \cap C)) \geq 1$, which yields $y(\delta(C)) \geq 1$.
  
  If $u_i = v_{i+1}$, then $x^i \in \PHK(B'_i, u_i, u_i)$, which means that $y(\delta(B'_i - C, B'_i \cap C)) \geq 2$, so $y(\delta(C)) \geq 1$ follows.
  
  If $C$ is a non-$s$-$t$-cut, we distinguish the following three cases.
  Note that in neither of the cases is $B'_0 \subseteq C$ or $B'_k \subseteq C$ a possibility, as that would make $C$ an $s$-$t$-cut.
  
  \textbf{If $C \subsetneq B'_i$ for some $i$}, and $C$ is a $u_i$-$v_{i+1}$-cut so that $u_i \in C$ (or $v_{i+1} \in C$), then $y(\delta(C))$ has at least $1$ load from the fact that $x^i \in \PHK(B'_i, u_i, v_{i+1})$, and $1$ load from the edge $v_{i-1} u_i$ (or $v_i u_{i+1}$).
  If $C$ is \textit{not} a $u_i$-$v_{i+1}$-cut, and $u_i, v_{i+1}$ are in $C$, then $y(\delta(C)) \geq 2$ because of the edges $v_i u_i$ and $v_{i+1} u_{i+1}$; while if $u_i, v_{i+1}$ are not in $C$ then $y(\delta(C)) \geq 2$ follows from $x^i \in \PHK(B'_i, u_i, v_{i+1})$.
  Note that if $u_i = v_{i+1}$, $C$ can only be a non-$u_i$-$v_{i+1}$-cut.
  In that case $x^i(\delta(C)) \geq 2$ because $x^i \in \PHK(B'_i, u_i, u_i)$, and thus $y(\delta(C)) \geq 2$ follows.
  
  \textbf{If $C = \bigcup\limits_{i \in \mathcal{I}} B'_i$} for some $\mathcal{I} \subset \{0, \dots, k\}$, let us define $\imin := \min\{i \mid B'_i \subset C\}$ and $\imax := \max\{i \mid B'_i \subset C\}$.
  Then $y(\delta(C)) \geq y(v_{\imin-1} u_{\imin}) + y(v_{\imax} u_{\imax+1}) = 2 \enspace$.
  
  \textbf{Else there exists a set $B'_i$ such that $C \cap B'_i \neq \emptyset$ and $C \nsubseteq B'_i$ and $B'_i \nsubseteq C$ hold}, then let us define $\imin$ and $\imax$ as follows:
  \begin{align*}
    \imin &:= \min\{i\mid C \cap B'_i\neq\emptyset, \, C \nsubseteq B'_i, \, B'_i \nsubseteq C \} \enspace , \\
    \imax &:= \max\{i\mid C \cap B'_i\neq\emptyset, \, C \nsubseteq B'_i, \, B'_i \nsubseteq C \} \enspace .
  \end{align*}  
  Suppose that $\imin \neq \imax$.
  Then, $\delta(C)$ has at least $1$ $y$-load on $\delta(B'_\imin \cap C, B'_\imin - C)$, as well as at least $1$ $y$-load between on $\delta(B'_\imax \cap C, B'_\imax - C)$, thus $y(\delta(C)) \geq 2$.
  In case of $\imin = \imax$, there must exist another index $i \neq \imin$ such that $B'_i \subset C$ (as otherwise we are back in one of the previous two cases), in which case there is at least one edge $e$ in $\delta(C \cap B'_i)$ such that $y(e) = 1$ (either $v_i u_i$ or $v_{i+1} u_{i+1}$ or both); in total $y(\delta(C)) \geq 2$ holds.

  This concludes the proof of \Cref{lem:polytopecontainment}.
\end{proof}
}

\iflong \prooflempolytopecontainment \else \seeinappendix \fi

\newcommand{\lemybgood}{
  $y$ is $\B$-good.
}

\begin{lemma}
\label{lem:ybgood}
  \lemybgood
\end{lemma}

\newcommand{\prooflemybgood}{
\begin{proof}
  The proof follows the lines of the corresponding proof of Lemma~3.3 of Zenklusen~\cite{Zenklusen2019}: there the claim can be deduced from cut constraints of $P_\text{HK}$, while in our case it follows from those of the polytope~$\PHK$.
  Nevertheless, we include the full proof here for the sake of completeness.

  For $i \in \{1, \dots, k\}$, we have by construction of $y$ that $y(v_i u_i)=1$ and $y(e)=0$ for other edges $e \in \delta(B_i)$.
  This means that all cuts $B_i$ satisfy (ii) of the definition of $\B$-goodness, i.e. the $y$-value is 1 and $y$ is integral.
  Let us show that for any other cut $B \in \B - \{B_1, \dots, B_k\}$, the $y$-load satisfies (i) of the definition.
  
  First suppose that $\{B_1, \dots, B_k\} \cup B$ is not a chain, in this case there is some index $j \in \{0, \dots, k\}$, such that neither $B \subseteq B_j$ nor $B_j \subseteq B$ is true.
  Hence
  \begin{align*}
    y(\delta(B)) + 1 &    = y(\delta(B)) + y(\delta(B_j)) \\
                     & \geq y(\delta(B - B_j)) + y(\delta(B_j - B)) \\
                     & \geq 4 \enspace .
  \end{align*}
  The first line follows from $y(\delta(B_j))=1$, this was shown at the beginning of the proof.
  The first inequality holds by the cut functions $C \rightarrow y(\delta(C))$ being symmetric and submodular.
  Since~$B$ and $B_j$ are $s$-$t$-cuts, $B - B_j$ and $B_j - B$ are non-$s$-$t$-cuts, and the $y$-load of both of these cuts is at least $2$, hence the second inequality follows.
  
  Suppose that $\{B_1, \dots, B_k\} \cup B$ is a chain, then there is an index $j$ such that $B_j \subsetneq B \subsetneq B_{j+1}$.
  If $u_j \in B$ and $v_{j+1} \notin B$, then $x^j(\delta(B)) \geq 3$ because of the constraints of the corresponding linear program $\LP(a)$, where $a = ((B_j, u_j), (B_{j+1}, v_{j+1}))$.
  Since $y \geq x^j$ holds for all $j$ component-wise, $y(\delta(B)) \geq 3$ follows.
  If $u_j \notin B$ and $v_{j+1} \in B$, then both the edges $v_j u_j$ and $v_{j+1} u_{j+1}$ are in~$\delta(B)$; moreover $x^j(\delta(B)) \geq 1$ since $B$ is a $u_j$-$v_{j+1}$-cut, therefore
  $$ y(\delta(B)) \geq x^j(\delta(B)) + y(v_j u_j) + y(v_{j+1} u_{j+1}) \geq 3 \enspace .$$
  
  Finally, if $B$ is not an $u_j$-$v_{j+1}$-cut, $x^j(\delta(B)) \geq 2$ since $x^j \in \PHK(B_{j+1} - B_j, u_j, v_{j+1})$.
  Moreover, either $u_i v_i$ or $u_{i+1} v_{i+1}$ is an edge in $\delta(B)$, depending on whether $u_i$ and $v_{j+1}$ are in~$B$ or not; both of the possibilities imply $y(\delta(B)) \geq 3$.
\end{proof}
}

\iflong \prooflemybgood \else \seeinappendix \fi

\newcommand{\lembgood}{
  Let $\B \subseteq \{C \subseteq V \mid s \in C, t \notin C \}$. One can determine in time polynomial in $|\B|$ and the input size of $(G, s, t, c)$ a $\B$-good point $y \in \PHK$ of minimum cost.
}

\begin{lemma}
\label{lem:bgood}
  \lembgood
\end{lemma}

\newcommand{\prooflembgood}{
\begin{proof}
  The number of nodes and arcs in $H$ are polynomial in $|\B|$.
  Calculating a shortest path on $H$ takes time polynomial in~$|H|$.
  The feasibility of the linear programs \cref{eq:LPa} can be checked in time $\poly(n, \log r, |\B|)$, therefore an optimal solution can also be found, using the ellipsoid method, in time polynomial in $n$, $\log r$ and $|\B|$ (see the discussion in \S 58.5 of Schrijver's book~\cite{Schrijver2003}).
\end{proof}
}

\newcommand{\remarksvtsp}{
\begin{remark}
It is worth considering how \Cref{alg:mvsttsp15} proceeds when applied to the single-visit TSP, that is, when $r(v)=1$ for each $v\in V$.
The output of \Cref{alg:gpolym} in \Cref{st:polym_alg} is then a connected multigraph with $r(V)-1=n-1$ edges.
This means that each vertex $v$ has degree at least $2\cdot r(v)-1 = 1$, which boils down to a connected graph with $n-1$ edges, therefore a spanning tree on $G$ with the additional properties \ref{req_Pq_1}-\ref{req_Pq_3}.
Thus \Cref{alg:mvsttsp15} performs the same operations, as the algorithm of Zenklusen~\cite{Zenklusen2019} for the {\sc Path TSP}.
\end{remark}
}

\iflong \prooflembgood \remarksvtsp \else \seeinappendix \fi

\iflong \else \newpage \fi

\section{Approximation Algorithm for the Bounded Degree g-Polymatroid Element with Multiplicities Problem}
\label{sec:approxpolymatroid}

\subsection{Polyhedral background}
\label{sec:bp}
In what follows, we make use of some basic notions and theorems of the theory of generalized polymatroids.
For background, see for example the paper of Frank and Tardos~\cite{FrankTardos1988} or Chapter~14 in the book by Frank~\cite{Frank2011}.

Given a ground set $S$, a set function $b:2^S\rightarrow\mathbbm{Z}$ is \emph{submodular} if
\begin{equation*}
  b(X)+b(Y)\geq b(X\cap Y) + b(X\cup Y)
\end{equation*} 
holds for every pair of subsets $X,Y\subseteq S$.
A set function $p:2^S\rightarrow\mathbbm{Z}$ is \emph{supermodular} if $-p$ is submodular.
As a generalization of matroid rank functions, Edmonds introduced the notion of polymatroids~\cite{Edmonds1970}.
A set function $b$ is a \emph{polymatroid function} if $b(\emptyset)=0$, $b$ is non-decreasing, and~$b$ is submodular.

We define
\begin{equation*}
  P(b):=\{x\in\mathbbm{R}^{S}_{\geq 0}\mid x(Y)\leq b(Y)\ \text{for every}\ Y\subseteq S\} \enspace.
\end{equation*}
The set of integral elements of $P(b)$ is called a \emph{polymatroidal set}.
Similarly, the \emph{base polymatroid}~$B(b)$ is defined by
\begin{equation*}
  B(b):=\{x\in\mathbbm{R}^{S}\mid x(Y)\leq b(Y)\ \text{for every}\ Y\subseteq S, \, x(S)=b(S)\} \enspace .
\end{equation*}
Note that a base polymatroid is just a facet of the polymatroid $P(b)$.
In both cases, $b$ is called the \emph{border function} of the polyhedron.
Although non-negativity of $x$ is not assumed in the definition of~$B(b)$, this follows by the monotonicity of $b$ and the definition of $B(b)$: $x(s)=x(S)-x(S-s) \geq b(S)-b(S-s)\geq 0$ holds for every $s\in S$.
The set of integral elements of~$B(b)$ is called a \emph{base polymatroidal set}.
Edmonds~\cite{Edmonds1970} showed that the vertices of a polymatroid or a base polymatroid are integral, thus $P(b)$ is the convex hull of the corresponding polymatroidal set, while $B(b)$ is the convex hull of the corresponding base polymatroidal set.
For this reason, we will call the sets of integral elements of $P(b)$ and $B(b)$ simply a polymatroid and a base polymatroid.

Hassin~\cite{Hassin1982} introduced polyhedra bounded simultaneously by a non-negative, monotone non-decreasing submodular function $b$ over a ground set $S$ from above and by a non-negative, monotone non-decreasing supermodular function $p$ over $S$ from below, satisfying the so-called \emph{cross-inequality} linking the two functions:
\begin{equation*}
  b(X) - p(Y) \geq b(X - Y) - p(Y - X)\qquad~\mbox{for every pair of subsets}~X,Y\subseteq S \enspace .
\end{equation*}
We say that a pair $(p,b)$ of set functions over the same ground set $S$ is a \emph{paramodular pair} if $p(\emptyset)=b(\emptyset)=0$, $p$ is supermodular, $b$ is submodular, and they satisfy the cross-inequality.
The slightly more general concept of generalized polymatroids was introduced by Frank~\cite{Frank1984}.
A \emph{generalized polymatroid}, or \emph{g-polymatroid} is a polyhedron of the form
\begin{equation*}
  Q(p,b):=\left\{ x \in \mathbbm{R}^{S}\mid p(Y) \leq x(Y) \leq b(Y) \ \text{for every}\ Y \subseteq S \right\} \enspace , 
\end{equation*}
where $(p,b)$ is a paramodular pair.
Here, $(p,b)$ is called the \emph{border pair} of the polyhedron. 
It is known (see e.g. \cite{Frank2011}) that a g-polymatroid defined by an integral paramodular pair is a non-empty integral polyhedron.

A special g-polymatroid is a box $\beta(L, U)=\{x\in \mathbbm{R} \sp {S}\mid  L\leq x\leq U\}$ where $L: S\rightarrow \mathbbm{Z} \cup \{-\infty \}$, $U: S\rightarrow \mathbbm{Z}\cup \{\infty \}$ with $L \leq U$. 
Another illustrious example is base polymatroids. 
Indeed, given a polymatroid function $b$ with finite $b(S)$, its \emph{complementary set function} $p$ is defined for $X\subseteq S$ by $p(X):=b(S)-b(S-X)$. 
It is not difficult to check that $(p,b)$ is a paramodular pair and that $B(b)=Q(p,b)$.

%The following result is stated in Theorem~14.3.9 of Frank~\cite{Frank2011}.

\begingroup
\def\thetheorem{14.3.9 (Frank~\cite{Frank2011})}
\begin{theorem}%[\textbf{restated}]
  The intersection $Q'$ of a g-polymatroid $Q=Q(p,b)$ and a box $\beta=\beta(L,U)$ is a g-polymatroid.
If $L(Y)\leq b(Y)$ and $p(Y)\leq U(Y)$ hold for every $Y\subseteq S$, then $Q'$ is non-empty, and its unique border pair $(p',b')$ is given by 
\begin{align}
\label{eq:gpolym_box}
\begin{split}
  p'(Z) &= \max\{ p(Z') - U(Z'-Z)+ L(Z-Z') \mid Z'\subseteq S\} \enspace , \\
  b'(Z) &= \min\{ b(Z') - L(Z'-Z)+ U(Z-Z') \mid  Z'\subseteq S\} \enspace .
\end{split}
\end{align} 
\end{theorem}
\addtocounter{theorem}{-1}
\endgroup

Given a g-polymatroid $Q(p, b)$ and $Z\subset S$, by \emph{deleting} $Z\subseteq S$ from $Q(p,b)$ we obtain a g-polymatroid $Q(p, b)\setminus Z$ defined on set $S - Z$ by the restrictions of~$p$ and $b$ to $S - Z$, that is,
\begin{equation*}
  Q(p, b)\setminus Z:=\{x\in\mathbbm{R}^{S-Z}\mid p(Y) \leq x(Y)\leq b(Y)\ \text{for every}\ Y\subseteq S - Z\} \enspace .
\end{equation*}
In other words, $Q(p, b)\setminus Z$ is the projection of $Q(p, b)$ to the coordinates in $S - Z$.

Extending the notion of contraction from matroids to g-polymatroids is not immediate.
A set can be naturally identified with its characteristic vector, that is, in the case of matroids contraction is basically an operation defined on $0{-}1$ vectors.
In our proof, we will need a generalization of this to the integral elements of a g-polymatroid.
However, such an element might have coordinates larger than one as well, hence finding the right definition is not straightforward. 
In the case of matroids, the most important property of contraction is the following: $I$ is an independent of $M/Z$ if and only if $F\cup I$ is independent in~$M$ for any maximal independent set~$F$ of~$Z$.
  
With this property in mind, we define the g-polymatroid obtained by the contraction of an integral vector $z\in Q(p,b)$ to be the polymatroid $Q(p',b'):=Q(p,b)/z$ on the same ground set~$S$ with the border functions
\begin{align*}
  p'(X) &:= p(X) - z(X) \\
  b'(X) &:= b(X) - z(X) \enspace .
\end{align*}
Observe that $p'$ is obtained as the difference of a supermodular and a modular function, implying that it is supermodular.
Similarly, $b'$ is submodular.
Moreover, $p'(\emptyset)=b'(\emptyset)=0$, and
\begin{align*}
  b'(X)-p'(Y)
  {}&{}=
  b(X)-z(X)-p(Y)+z(Y)\\
  {}&{}\geq
  b(X-Y)+p(Y-X)-z(X-Y)+z(Y-X)\\
  {}&{}=
  b'(X-Y)-p'(Y-X),
\end{align*}
hence $(p',b')$ is indeed a paramodular pair.
The main reason for defining the contraction of an element $z\in Q(p,b)$ is shown by the following lemma.
\begin{lemma}
\label{lem:contraction}
  Let $Q(p',b')$ be the polymatroid obtained by contracting $z\in Q(p,b)$.
  Then $x+z\in Q(p,b)$ for every $x\in Q(p',b')$.
\end{lemma}
\begin{proof}
  Let $x\in Q(p',b')$. By definition, this implies $p'(Y)\leq x(Y)\leq b'(Y)$ for $Y\subseteq S$. 
  Thus $p(Y)=p'(Y)+z(Y)\leq x(Y)+z(Y)\leq b'(Y)+z(Y)=b(Y)$, concluding the proof.
\end{proof}

%Formally, the {\sc Bounded Degree g-po\-ly\-mat\-ro\-id Element Problem} takes as input a g-polymatroid $Q(p,b)$ with a cost function $c:S \rightarrow \mathbbm{R}$, and a hypergraph $H=(S, \Eps)$ on the same ground set with lower and upper bounds $f,g:\Eps\rightarrow\mathbbm{Z}_{\geq 0}$ and multiplicity vectors $ m_\eps:S\rightarrow\mathbbm{Z}_{\geq0}$ for $\eps\in\Eps$ satisfying $m_\eps(s)=0$ for $s\in S-\eps$.
%The objective is to find a minimum-cost element $x$ of $Q(p,b)$ such that $f(\eps) \leq \sum_{s\in \eps} m_\eps(s) x(s) \leq g(\eps)$ for each $\eps \in \Eps$.
  
\subsection{The algorithm}

The aim of this section is to prove \Cref{thm:matroid1} and \Cref{thm:matroid2}.
\Cref{thm:matroid1} extends the result of Kir{\'a}ly et al.~\cite{KiralyLS2012} from matroids to g-poly\-matroids.
%It turns out that, when upper bounds are present, there is a significant difference when g-polymatroids are considered instead of matroids.
However, adapting their algorithm is not immediate due to the following major differences.
A crucial step of their approach is to relax the problem by deleting a constraint corresponding to a hyperedge~$\eps$ with small $g(\eps)$ value.
This step is feasible when the solution is a $0{-}1$ vector, 
%as in those cases the violation on $\eps$ is upper bounded by the size of the hyperedge.
%This but does not hold 
but it is not applicable for g-polymatroids (or even for polymatroids) where an integral element might have coordinates larger than 1.
This difficulty is compounded by the presence of multiplicity vectors, that makes both the computations and the tracking of changes after hyperedge deletions more complicated.
Finally, in contrast to matroids that are defined by a submodular function (the rank function), g-polymatroids are determined by a pair of supermodular and submodular functions.
Thus the structure of the family of tight sets is more complex, which affects the proof of one of the key claims (\Cref{claim:uncrossing}).
%However, we show that after the first round of our algorithm, the problem can be restricted to the unit cube and so upper bounds remain tractable. 

We start by formulating a linear programming relaxation for the {\sc Bounded Degree g-po\-ly\-mat\-roid Element Problem}:

\leqnomode
\begin{align*}
  \label[empty]{eq:lp_poly}
  \text{minimize} \qquad \sum_{s \in S}  c(s) \ &x(s) \\
  \text{subject to} \qquad p(Z)  \leq \ &x(Z) \leq b(Z) &\forall Z \subseteq S \tag{LP}\\
  \qquad f(\eps) \leq \sum_{s\in \eps} \ m_\eps(s) \, &x(s) \leq g(\eps) &\forall \eps \in \Eps %\\
%  \qquad &x(s) \geq 0 &\forall s\in S
\end{align*}

Although the program has an exponential number of constraints, it can be separated in polynomial time using submodular minimization \cite{IwataFF2001,McCormick2005,Schrijver2000}.
\Cref{alg:gpolym} generalizes the approach by Kir{\'a}ly et al.~\cite{KiralyLS2012}.
We iteratively solve the linear program, delete elements which get a zero value in the solution, update the solution values and perform a contraction on the polymatroid, or remove constraints arising from the hypergraph.
In the first round, the bounds on the coordinates solely depend on $p$ and $b$, while in the subsequent rounds the whole problem is restricted to the unit cube.

\begin{algorithm}[h!]
\caption{Approximation algorithm for the {\sc Bounded Degree g-poly\-matroid Element with Multiplicities} problem.}\label{alg:gpolym}
\begin{algorithmic}[1]
  \Statex \textbf{Input:} A g-polymatroid $Q(p,b)$ on ground set $S$, cost function $c:S\rightarrow\mathbbm{R}$, a hypergraph $H = (S, \Eps)$, lower and upper bounds $f,g:\Eps\rightarrow\mathbbm{Z}_{\geq 0}$, multiplicities $m_\eps:S\rightarrow\mathbbm{Z}_{\geq 0}$ for $\eps \in \Eps$ satisfying $m_\eps(s)=0$ for $s\in S-\eps$.
  \Statex \textbf{Output:} $z\in Q(p,b)$ of cost at most $\textsc{OPT}_{LP}$, violating the hyperedge constraints by at most $2\Delta-1$. 
  \State Initialize $z(s) \leftarrow 0$ for every $s\in S$.
  \While{$S\neq\emptyset$}
    \State \begin{varwidth}[t]{0.9\linewidth} Compute a basic optimal solution $x$ for \Cref{eq:lp_poly}. \\ (Note: starting from the second iteration, $0\leq x \leq 1$.) \end{varwidth}
    \begin{algsubstates}
      \State \begin{varwidth}[t]{0.9\linewidth}
      Delete any element $s$ with $x(s)=0$.
      Update each hyperedge $\eps\leftarrow \eps-s$ and $m_\eps(s)\leftarrow 0$.
      Update the g-polymatroid $Q(p,b)\leftarrow Q(p,b)\setminus s$ by deletion. \label{st:del}
     	\end{varwidth}
     	\State \begin{varwidth}[t]{0.9\linewidth}
      For all $s\in S$ update $z(s) \leftarrow z(s) + \floor{x}(s)$.\\
      Apply polymatroid contraction $Q(p,b)\leftarrow Q(p,b)/\floor{x}$, that is, redefine $p(Y) := p(Y) - \floor{x}(Y)$ and $b(Y) := b(Y) - \floor{x}(Y)$ for every $Y \subseteq S$.\\
      Update $f(\eps) \leftarrow f(\eps) - \displaystyle\sum_{s\in \eps}\, m_\eps(s) \floor{x}(s)$ and $g(\eps) \leftarrow g(\eps) - \displaystyle\sum_{s\in \eps}\, m_\eps(s) \floor{x}(s)$ for each $\eps \in \Eps$.\label{st:inc}
      \end{varwidth}
    	\State If $m_\eps(\eps) \leq 2\Delta-1$, let $\Eps \leftarrow \Eps - \eps$. \label{st:rem}
      \State \begin{varwidth}[t]{0.9\linewidth} \textbf{if} it is the first iteration \textbf{then} \label{st:first}\\
      Take the intersection of $Q(p,b)$ and the unit cube $[0,1]^S$, that is, $p(Y):=\max\{ p(Y') - |Y'-Y| \mid Y'\subseteq S\}$ and $b(Y) := \min\{ b(Y')+ |Y-Y'| \mid  Y'\subseteq S\}$ for every $Y\subseteq S$.
      \end{varwidth}    	    	 
      \end{algsubstates}
    \EndWhile
  \State \textbf{return} $z$
\end{algorithmic}
\end{algorithm}

\thmrestate{thm:matroid1}{\thmmatroidtwoside}{theorem}

\begin{proof}%[Proof of \Cref{thm:matroid1}]

Our algorithm is presented as \Cref{alg:apx_tp_path}.
   \paragraph{\textbf{Correctness.}}
  First we show that if the algorithm terminates then the returned solution $z$ satisfies the requirements of the theorem.  
  In a single iteration, the g-polymatroid $Q(p,b)$ is updated to $(Q(p,b)\setminus D)/\floor{x}$, where \mbox{$D=\{s:x(s)=0\}$} is the set of deleted elements.
  In the first iteration, the g-polymatroid thus obtained is further intersected with the unit cube.
  By \Cref{lem:contraction}, the vector $x-\lfloor x\rfloor$ restricted to $S-D$ remains a feasible solution for the modified linear program in the next iteration.
  Note that this vector is contained in the unit cube as its coordinates are between $0$ and $1$.
  This remains true when a lower degree constraint is removed in \Cref{st:rem} as well, therefore the cost of $z$ plus the cost of an optimal LP solution does not increase throughout the procedure.
  Hence the cost of the output~$z$ is at most the cost of the initial LP solution, which is at most the optimum.
  
  By \Cref{lem:contraction}, the vector $x-\lfloor x \rfloor+z$ is contained in the original g-polymatroid, although it might violate some of the lower and upper bounds on the hyperdeges.
  We only remove the constraints corresponding to the lower and upper bounds for a hyperedge $\eps$ when $m_\eps(\eps) \leq 2\Delta-1$.
  As the g-polymatroid is restricted to the unit cube after the first iteration, these constraints are violated by at most $2\Delta-1$, as the total value of $\sum_{s\in\eps}m_\eps(s)z(s)$ can change by a value between~$0$ and $2\Delta-1$ in the remaining iterations.
 
  It remains to show that the algorithm terminates successfully.
  The proof is based on similar arguments as in Kir{\'a}ly et al.~\cite[proof of Theorem 2]{KiralyLS2012}.

  \paragraph{\textbf{Termination.}}
  Suppose, for sake of contradiction, that the algorithm does not terminate.
  Then there is some iteration after which none of the simplifications in \Crefrange{st:del}{st:rem} can be performed.
  This implies that for the current basic LP solution $x$ it holds $0<x(s)<1$ for each $s \in S$ and $m_\eps(\eps)\geq 2\Delta$ for each $\eps \in \Eps$.
  We say that a set $Y$ is \emph{p-tight} (or \emph{b-tight}) if $x(Y)=p(Y)$ (or $x(Y)=b(Y)$), and let $\varT^p=\{Y\subseteq S:x(Y)=p(Y)\}$ and $\varT^b=\{Y\subseteq S:x(Y)=b(Y)\}$ denote the collections of $p$-tight and $b$-tight sets with respect to solution $x$.
  
  Let $\varL$ be a maximal independent laminar system in $\varT^p \cup \varT^b$.
  \begin{claim}
  \label{claim:uncrossing}
    $\spa{(\{\chi_Z\mid Z \in \varL\})} = \spa{(\{\chi_Z\mid Z \in \varT^p \cup \varT^b\})}$
  \end{claim}
  \begin{proof}[Proof of \Cref{claim:uncrossing}.]
  \renewcommand{\qedsymbol}{$\Diamond$}
    The proof uses an uncrossing argument.
    Let us suppose indirectly that there is a set $R$ from~$\varT^p \cup \varT^b$ for which $\chi_R \notin \spa{(\{\chi_Z\mid Z \in \varL\})}$.
    Choose this set $R$ so that it is incomparable to as few sets of $\varL$ as possible.
    Without loss of generality, we may assume that $R \in \varT^p$.
    Now choose a set $T \in \varL$ that is incomparable to $R$.
    Note that such a set necessarily exists as the laminar system is maximal.
    We distinguish two cases.
  
    \noindent \textbf{Case 1.} $T \in \varT^p$.
    Because of the supermodularity of $p$, we have
    \begin{align*}
      x(R) + x(T) &= p(R) + p(T) \leq p(R \cup T) + p(R \cap T) \leq x(R \cup T) + x(R \cap T)\\
                  &= x(R) + x(T),
    \end{align*}
    hence equality holds throughout.
    That is, $R \cup T$ and $R \cap T$ are in $\varT^p$ as well.
    In addition, since $\chi_R + \chi_T = \chi_{R \cup T} + \chi_{R \cap T}$ and $\chi_R$ is not in $\spa{(\{\chi_Z\mid Z \in \varL\})}$, either $\chi_{R \cup T}$ or $\chi_{R \cap T}$ is not contained in $\spa{(\{\chi_Z\mid Z \in \varL\})}$.
    However, both $R \cup T$ and $R \cap T$ are incomparable with fewer sets of $\varL$ than~$R$, which is a contradiction.
  
    \noindent \textbf{Case 2.} $T \in \varT^b$.
    Because of the cross-inequality, we have
    \begin{align*}
      x(T) - x(R) &= b(T) - p(R) \geq b(T \setminus R) - p(R \setminus T) \geq x(T \setminus R) - x(R \setminus T)\\
                  &= x(T) - x(R),
    \end{align*}
    implying $T \setminus R \in \varT^b$ and $R \setminus T \in \varT^p$.
    Since $\chi_R + \chi_T = \chi_{R \setminus T} + \chi_{R \setminus T} + 2 \ \chi_{R \cup T}$ and $\chi_R$ is not in $\spa{(\{\chi_Z\mid Z \in \varL\})}$, one of the vectors $\chi_{R \setminus T}$, $\chi_{R \setminus T}$ and $\chi_{R \cup T}$ is not contained in $\spa{(\{\chi_Z\mid Z \in \varL\})}$.
    However, any of these three sets is incomparable with fewer sets of $\varL$ than~$R$, which is a contradiction.
  
    The case when $R \in \varT^b$ is analogous to the above.
    This completes the proof of the Claim.
  \end{proof}
  
  We say that a hyperedge $\eps \in \Eps$ is \emph{tight} if $f(\eps)=\sum_{s\in \eps} m_\eps(s) x(s)$ or $g(\eps)=\sum_{s\in \eps} m_\eps(s) x(s)$. % m_\eps(s)
  As $x$ is a basic solution, there is a set $\Eps'\subseteq\Eps$ of tight hyperedges such that $\{m_\eps\mid \eps \in \Eps'\}\cup\linebreak \{\chi_Z\mid Z \in \varL\}$ are linearly independent vectors with $|\Eps'|+|\varL|=|S|$.
  
  We derive a contradiction using a token-counting argument.
  We assign $2\Delta$ tokens to each element $s \in S$, accounting for a total of $2\Delta |S| $ tokens.
  The tokens are then redistributed in such a way that each hyperedge in $\Eps'$ and each set in $\varL$ collects at least $2\Delta$ tokens, while at least one extra token remains.
  This implies that $2\Delta |S|>2\Delta|\Eps'|+2\Delta|\varL|$, leading to a contradiction.
  
  We redistribute the tokens as follows.
  Each element $s$ gives $\Delta$ tokens to the smallest member in~$\varL$ it is contained in, and $m_\eps(s)$ tokens to each hyperedge $\eps\in\Eps'$ it is contained in.
  As \mbox{$\sum_{\eps\in\Eps:s\in \eps} m_\eps(s)\leq\Delta$} holds for every element $s \in S$, thus we redistribute at most $2\Delta$ tokens per element and so the redistribution step is valid.  
  Now consider any set $U\in\varL$. 
  Recall that $\varL^{\max}(U)$ consists of the maximal members of $\varL$ lying inside $U$. 
  Then $U-\bigcup_{W\in\varL^{\max}(U)} W\neq\emptyset$, as otherwise $\chi_U=\sum_{W\in\varL^{\max}(U)} \chi_W$, contradicting the independence of $\varL$.
  For every set $Z$ in~$\varL$, $x(Z)$ is an integer, meaning that $x(U - \bigcup_{W\in\varL^{\max}(U)} W)$ is an integer.  
  But also $0 < x(s) < 1$ for every $s \in S$, which means that $U - \bigcup_{W\in\varL^{\max}(U)} W$ contains at least 2 elements.
  Therefore, each set $U$ in $\varL$ receives at least $2 \Delta$ tokens, as required.
  By assumption, $m_\eps(\eps) \geq 2 \Delta$ for every hyperedge $\eps \in \Eps'$, which means that each hyperedge in $\Eps'$ receives at least $2\Delta$~tokens, as required. 

  If $\sum_{\eps\in\Eps':s\in \eps} m_\eps(s)\leq \Delta$ holds for any $s\in S$ or $\mathcal{L}^{\max}(S)$ is not a partition of $S$, then an extra token exists.
  Otherwise, $\sum_{\eps\in\Eps'}m_{\eps} = \Delta\cdot\chi_S = \Delta \cdot \sum_{W\in\varL^{\max}(S)} \chi_W$, contradicting the independence of $\{m_\eps\mid \eps \in \Eps'\}\cup \{\chi_Z\mid Z \in \varL\}$.
     
  \paragraph{\textbf{Time complexity.}}
  Solving an LP, as well as removing a hyperedge in \Cref{st:del} or removing an element from a hyperedge in \Cref{st:rem} can be done in polynomial time.
  In \Cref{st:inc,st:first}, we calculate the value of the current functions~$p$ and~$b$ for a set $Y$ only when it is needed during the ellipsoid method.
  We keep track of the vectors $\floor{x}$ that arise during contraction steps (there is only a polynomial number of them), and every time a query for $p$ or $b$ happens, it takes into account every contraction and removal that occurred until that point.

  \Cref{st:del} can be repeated at most $|S|$ times, while \Cref{st:rem} can be repeated at most $|\Eps|$ times.
  Starting from the second iteration, we are working in the unit cube. 
  That is, when \Cref{st:inc} adds the integer part of a variable $x(s)$ to $z(s)$ and reduces the problem, then the given variable will be~$0$ in the next iteration and so element $s$ is deleted. 
  This means that the total number of iterations of \Cref{st:inc} is at most $\O(|S|)$.

\end{proof}

Now we consider case when only lower or only upper bounds are given.

\thmrestate{thm:matroid2}{\thmmatroidoneside}{theorem}

\begin{proof}%[Proof of \Cref{thm:matroid2}]
  The proof is similar to the proof of \Cref{thm:matroid1}, the main difference appears in the counting argument.
  When only lower bounds are present, the condition in \Cref{st:rem} changes: we delete a hyperedge $\eps$ if $f(\eps)\leq\Delta-1$.
  Suppose, for the sake of contradiction, that the algorithm does not terminate.
  Then there is an iteration after which none of the simplifications in \Crefrange{st:del}{st:rem} can be performed.
  This implies that in the current basic solution $0 < x(s) < 1$ holds for each $s \in S$ and $f(\eps) \geq \Delta$ for each $\eps \in \Eps$.
  We choose a subset $\Eps'\subseteq\Eps$ and a maximal independent laminar system~$\varL$ of tight sets the same way as in the proof of \Cref{thm:matroid1}.
  Recall that $|\Eps'| + |\varL| = |S|$.
  
  Let $Z_1, \dots, Z_k$ denote the members of the laminar system $\varL$. 
  As $\varL$ is an independent system, $Z_i-\bigcup_{W\in\mathcal{L}^{\max}(Z_i)}W\neq\emptyset$.
  Since $x(s)<1$ for all $s\in S$, $$x(Z_i-\bigcup_{W\in\mathcal{L}^{\max}(Z_i)}W)<|Z_i-\bigcup_{W\in\mathcal{L}^{\max}(Z_i)}W| \enspace .$$
  As we have integers on both sides of this inequality, we get 
  \begin{equation*}
    |Z_i-\!\!\!\bigcup_{W\in\mathcal{L}^{\max}(Z_i)}\!\!\!\!\!\!W|-x(Z_i-\!\!\!\bigcup_{W\in\mathcal{L}^{\max}(Z_i)}\!\!\!\!\!\!W)\geq 1\quad\text{for all}\ i=1,\dots,k \enspace .
  \end{equation*}
  Moreover, $\sum_{s\in\eps}m_{\eps}(s)x(s)\geq f(\eps)\geq\Delta$ for all hyperedges; therefore,
  \begin{align*}
    |\Eps'| + |\varL|
    {}&{}\leq 
    \sum_{\eps \in \Eps'} \frac{\sum_{s \in \eps} m_\eps(s) x(s)}{\Delta} + \sum_{i=1}^k \left[ |Z_i - \!\!\! \bigcup_{W \in \mathcal{L}^{\max}(Z_i)} \!\!\!\!\!\! W| - x(Z_i - \!\!\! \bigcup_{W \in \mathcal{L}^{\max}(Z_i)} \!\!\!\!\!\! W) \right] \\
    {}&{}= 
    \sum_{s \in S} \frac{x(s)}{\Delta} \sum_{\substack{\eps \in \Eps' \\ s\in \eps}} m_\eps(s) + \sum_{W \in \mathcal{L}^{\max}(S)}|W| - \sum_{W \in \mathcal{L}^{\max}(S)} x(W) \leq |S| \enspace . \label{eq:optional} 
  \end{align*}
  In the last line, the first term is at most $x(S)$ since $\sum_{\eps\in\Eps:s\in \eps} m_\eps(s)\leq\Delta$ holds for each element~\mbox{$s \in S$}.
  From $x(S)- \sum_{W \in \mathcal{L}^{\max}(S)} x(W)\leq |S|-\sum_{W \in \mathcal{L}^{\max}(S)}|W|$ the upper bound of $|S|$ follows.
  As $|S| = |\varL| + \mathcal{|\Eps'|}$, we have equality throughout. 
  This implies that
  \begin{equation*}
  \sum_{\eps \in \Eps'} m_\eps = \Delta \cdot \chi_S=\Delta\cdot\sum_{W\in\mathcal{L}^{\max}(S)}\chi_W,
  \end{equation*}
  contradicting linear independence.
  
  If only upper bounds are present, we remove a hyperedge $\eps$ in \Cref{st:rem} when $g(\eps)+\Delta-1 \geq m_\eps(\eps)$. 
  Suppose, for the sake of contradiction, that the algorithm does not terminate.
  Then there is an iteration after which none of the simplifications in \Crefrange{st:del}{st:rem} can be performed.
  This implies that in the current basic solution $0 < x(s) < 1$ holds for each $s \in S$ and $m_\eps(\eps)-g(\eps) \geq \Delta$ for each $\eps \in \Eps$.
  Again, we choose a subset $\Eps'\subseteq\Eps$ and a maximal independent laminar system~$\varL$ of tight sets the same way as in the proof of \Cref{thm:matroid1}. 
     
  Let $Z_1, \dots, Z_k$ denote the members of the laminar system $\varL$. 
  As $\varL$ is an independent system, $Z_i-\bigcup_{W\in\mathcal{L}^{\max}(Z_i)}W\neq\emptyset$ and so
  \begin{equation*}
    x(Z_i - \!\!\! \bigcup_{W \in \mathcal{L}^{\max}(Z_i)} \!\!\!\!\!\! W) \geq 1 \enspace .
  \end{equation*}
  By $\sum_{s \in \eps} m_\eps(s) x(s) \leq g(\eps)$, we get $\sum_{s \in \eps} m_\eps(s)-\sum_{s \in \eps} m_\eps(s) x(s) \geq m_\eps(\eps)-g(\eps) \geq \Delta$. 
  Thus,
  \begin{align*}
    |\Eps'| + |\varL|
    {}&{}\leq
    \sum_{\eps \in \Eps'} \frac{\sum_{s \in \eps} m_\eps(s)-\sum_{s \in \eps} m_\eps(s) x(s)}{\Delta} + \sum_{i=1}^k x(Z_i - \!\!\! \bigcup_{W \in \mathcal{L}^{\max}(Z_i)} \!\!\!\!\!\! W) \\
    {}&{}=
    \sum_{s \in S} \frac{1-x(s)}{\Delta} \sum_{\substack{\eps \in \Eps' \\ s\in \eps}} m_\eps(s) + \sum_{W \in \mathcal{L}^{\max}(S)} x(W) \label{eq:optional2} \\
    {}&{}\leq
    \sum_{s \in S} \frac{1-x(s)}{\Delta} \sum_{\substack{\eps \in \Eps' \\ s\in \eps}} m_\eps(s) +  x(S) \leq |S| \enspace .
  \end{align*} 
  In the last line, the first term is at most $|S|-x(S)$ since $\sum_{\eps\in\Eps:s\in \eps} m_\eps(s)\leq\Delta$ holds for every element $s \in S$.
  Therefore, the upper bound of $|S|$ follows.
  As $|S| = |\varL| + \mathcal{|\Eps'|}$, we have equality throughout. 
  %This implies that $\sum m_\eps(s) = \Delta$ for $s\in S$ and necessarily $Z_k=S$.
  This implies that $\sum_{\eps \in \Eps'} m_\eps = \Delta \cdot \chi_S=\Delta\cdot\sum_{W\in\mathcal{L}^{\max}(S)}\chi_W$, contradicting linear independence.
\end{proof}

\begin{remark}
Note that \Cref{thm:matroid1,thm:matroid2} only provide a solution if there exists a (fractional) solution to the underlying linear program in \Cref{eq:lp_poly}.
Consequently, \Cref{thm:bdmultigraph} only provides a solution if the polytope $\PSG$ in \Cref{eq:pcg_general} is not empty.
\end{remark}

We have seen in \Cref{sec:bp} that base polymatroids are special cases of g-polymatroids.
This implies that the results of \Cref{thm:matroid2} immediately apply to polymatroids.
In the {\sc Lower Bounded Degree Polymatroid Basis with Multiplicities} problem, we are given a base polymatroid $B(b)=(S,b)$ with a cost function $c:S \rightarrow \mathbbm{R}$, and a hypergraph $H=(S, \Eps)$ on the same ground set.
The input contains lower bounds $f: \Eps \rightarrow \mathbbm{Z}_{\geq 0}$ and multiplicity vectors $m_\eps: \eps \rightarrow \mathbbm{Z}_{\geq 1}$ for every hyperedge $\eps \in \Eps$.
The objective is to find a minimum-cost element $x \in B(b)$ such that $f(\eps) \leq \sum_{s \in \eps} m_\eps(s) x(s)$ holds for each $\eps \in \Eps$.

\begin{corollary}
\label{thm:polym}
  There is a polynomial-time algorithm for the {\sc Lower Bounded Degree Polymatroid Basis with Multiplicities} problem which returns an integral element $x$ of $B(b)$ of cost at most the optimum value such that $f(\eps)- \Delta+1 \leq \sum_{s \in \eps} m_\eps(s) x(s)$ for each $\eps \in\Eps$.
\end{corollary}  

\subsection{Proof of \Cref{thm:bdmultigraph}}
\label{sec:approx}

In this section we show that \Cref{alg:gpolym} can be applied in order to obtain an approximation to the {\sc Minimum Bounded Degree Connected Multigraph with Edge Bounds} problem, as described in \Cref{thm:bdmultigraph}.

\begingroup
\def\thetheorem{\ref*{thm:bdmultigraph}}
\begin{theorem}%[\textbf{restated}]
  \thmbdmultigraph 
  \begin{equation*}
    \PSG(\rho, L, U) := \left\{ x \in \mathbbm{R}^E_{\geq 0} \midbar \begin{array}{ll}
    \supp(x) \text{ is connected} \\
    x(E) = \nicefrac{\rho(V)}{2} \\ \tag{\ref*{eq:pcg_general}}
    x(\ddelta(v)) \geq \rho(v)  \qquad \qquad \quad \qquad \forall v \in V \\
    L(vw)\leq x(vw) \leq U(vw) \qquad \forall v,w \in V
    \end{array} \right\} \enspace .
  \end{equation*}
\end{theorem}
\addtocounter{theorem}{-1}
\endgroup

Let us take a {\sc Minimum Bounded Degree Connected Multigraph with Edge Bounds} problem instance $(G,c,\rho,L,U)$ on a graph $G(V,E)$, where $c$, $\rho$, $L$, $U$ are non-negative and $\rho(V) = \sum_{v \in V} \rho(v)$ is even.\!\!\!
\footnote{Due to the handshaking lemma, the sum of degrees in a graph is even, therefore $\rho(V)$ being even is necessary.}
Note that we do not require $c$ to satisfy the triangle inequality.
We start with defining the specific input variables passed over \Cref{alg:gpolym}.
Then, we show that given the specified input, the algorithm yields an approximate solution to the {\sc Minimum Bounded Degree Connected Multigraph  with Edge Bounds} problem.
From now on we use $\hat{\rho}=\nicefrac{\rho(V)}{2}-|V|+1$.

We first set the base set $S$ as the edge set $E$ of our original graph $G$.
In the hypergraph $H=(S,\Eps)$, the elements of $S$ thus correspond to the edges of $G$.
Moreover, there is a hyperedge~$\eps$ for every vertex in $V$, defined the following way: $\Eps := \{ \delta(v)\mid v \in V \}$.
The multiplicity of an element $s$ in a hyperedge $\eps$ is 1, that is, $m_\eps(s) := 1$ if $s$ corresponds to a regular edge $e \in E$, and $m_\eps(s) := 2$ if $s$ corresponds to a self-loop.
We set the lower bound $f$ for a hyperedge $\eps$ according to the degree requirement of the corresponding vertex $v$, that is $f(\eps) := \rho(v)$.

We now define the second input of \Cref{alg:gpolym}, a g-polymatroid $Q(S,p,b)$.
This is done in two steps, by first defining an auxiliary polymatroid $Q'(S,p',b')$, then taking the intersection of the g-polymatroid $Q'$ with a box.
We define the border function $p'$ as the zero vector on $S$, and $b'(Z)$ as follows:

\begin{lemma}
\label{lem:bdef}
  Let $b'$ denote the following function defined on sets $Z\subseteq S$:
  \begin{equation}
  \label{eq:b}
    b'(Z) = \begin{cases}
             |V(Z)|-\comp(Z)+\hat{\rho}, & \text{if $Z\neq\emptyset$,}\\
              0, & \text{otherwise \enspace .}
           \end{cases}
  \end{equation} 
  Then $b'$ is a polymatroid function.
\end{lemma}
\begin{proof}
  By definition, $b'(\emptyset)=0$ and $b$ is monotone increasing.
  It remains to show that $b'$ is submodular.
  Let $X,Y\subseteq S$.
  The submodular inequality clearly holds if one of $X$ and $Y$ is empty. If none of $X$~and~$Y$ is empty then the submodular inequality follows from the fact that $|V(Z)|-\comp(Z)$ is the rank function of the graphical matroid. 
\end{proof}

Consider the g-polymatroid $B(p',b')$ determined by the border functions defined in \Cref{eq:b}.
Let us define the set $B=\{x\in\mathbbm{Z}^{E}_{\geq 0}:x(E)=\nicefrac{\rho(V)}{2}, ~ \supp(x)\ \text{is connected}\}$.

\begin{lemma}
\label{lem:description}
  $B=B(p',b')\cap\mathbbm{Z}^{E}_{\geq 0}$.
\end{lemma}
\begin{proof}
  Take an integral element $x\in B(p',b')$ and let $C\subseteq E$ be an arbitrary cut between $V_1$ and~$V_2$ for some partition $V_1\uplus V_2$ of $V$.
  Then
  \begin{align*}
    x(C) {}&{}= x(E)-x\big(E(V_1)\cup E(V_2)\big)\\
         {}&{}\geq |V|-1+\hat{\rho}-(|V_1|+|V_2|-\comp\!\big(E(V_1)\cup E(V_2)\big)+\hat{\rho})\\
         {}&{}\geq 1,
  \end{align*}
  thus $\supp(x)$ is connected.
  As $x(E)=|V|-1+\hat{\rho}=\nicefrac{\rho(V)}{2}$, we obtain $x\in B$, showing that $B(p',b')\subseteq B$.

  To see the other direction, take an element $x\in B$.
  As $\supp(x)$ is connected, $x(E - F)\geq\comp(F)+|V|-|V(F)|-1$ for every $F\subseteq E$.
  That is, 
  \begin{align*}
    x(F) {}&{}   = x(E) - x(E - F)\\
         {}&{}\leq r(V)-(|V-V(F)| +\comp(F)-1)\\
         {}&{}   = |V(F)|-\comp(F)+\hat{\rho},
  \end{align*}
  thus $x(F)\leq b'(F)$.
  As $x(E) = r(V) = |V|-1+\hat{\rho}$, we obtain $x\in B(p',b')$, showing $B\subseteq B(p',b')$. 
\end{proof}

So far we proved that the integral points of $Q'(S,p',b')$ correspond to a connected multigraph on $V$ that has $\nicefrac{\rho(V)}{2}$ edges.
Let $\beta(L, U)$ be the box defined by $$\beta(L,U) := \left\{x \in \mathbbm{R}^S_{\geq 0} \mid L(s) \leq x(s) \leq U(s) \enspace \forall s \in S \right\} \enspace .$$
Let us define the polymatroid $Q=(S,p,b)$ as the intersection of the polymatroid $Q'$ and the box~$\beta$, where the border functions $p, b$ are defined as in \Cref{eq:gpolym_box}.
We now prove that taking $H=(S,\Eps)$ and $Q(S,p,b)$ as input, the output of \Cref{alg:gpolym} corresponds to a multigraph with the properties stated in \Cref{thm:bdmultigraph}.

\begin{proof}[Proof of \Cref{thm:bdmultigraph}]

Consider the linear program~\Cref{eq:lp_poly} that is defined in the iterative rounding method for the \textsc{g-Polymatroid Element with Multiplicities} problem.
The constraints regarding the bounds on the hyperedges imply $\rho(v) \leq x(\ddelta(v))$ for every $v \in V$:
note that $m_\eps(s)=2$ for self loops and $1$ for simple edges, and this equals to the contribution of an edge~$uv$ to the value $x(\ddelta(v))$.
This, together with \Cref{lem:description} and the fact that $x$ is contained in the box $\beta(L, U)$, implies that \Cref{alg:gpolym} returns an integral solution~$z$ such that the cost of $z$ is at most the minimum cost element of $\PSG$.

According to \Cref{thm:matroid2}, the integral solution $z$ violates the bounds $f$ on the hyperedges by at most $\Delta-1$, where $\Delta := \max_{s \in S} \left\{ \sum_{\eps \in \Eps: s \in \eps} m_\eps(s) \right\}$.
But we defined $m_\eps(s)$ to be equal to~$2$ if $s$ corrensponds to a self-loop in $G$ and $1$ otherwise, meaning that the solution $z$ violates the bounds on the hyperedges $f$ and thus the bounds on the vertices $\rho$ by at most 1.
The solution~$z$ is also connected, with a total number of edges $\nicefrac{\rho(V)}{2}$ and $z$ satisfies the edge bounds $L, U$; due to \Cref{thm:matroid2}.
Therefore, the solution $z$ corresponds to a multigraph, that admits the properties in the claim of \Cref{thm:bdmultigraph}.
\end{proof}

\subsection{Strongly polynomial time implementation}

The transportation problem in \Cref{alg:apx_tp_path} can be solved in strongly polynomial time~\cite{Orlin1993,KleinschmidtSchannath1995}.
Computing compact path-cycle representations uses the algorithm of Grigoriev and van de Klundert~\cite{Grigoriev2006}; which, along with computing the Eulerian trail and making shortcuts, can be done in strongly polynomial time.
Moreover, the algorithm uses the $\nicefrac32$-approximation for the {\sc Path TSP} by Zenklusen~\cite{Zenklusen2019} as a black-box, which can also be implemented in strongly polynomial time.
Making shortcuts in \Cref{alg:apx_tp_path} can be performed in strongly polynomial time as well, thus we can find a $\nicefrac52$-approximation for the metric {\sc Many-visits Path TSP} in strongly polynomial time.

\Cref{alg:mvsttsp15} involves solving three types of LPs.
According to \S 58.5 in Schrijver's book~\cite{Schrijver2003}, if the feasibility of a linear program for a vector $x$ can be tested in polynomial time, then the ellipsoid method can find a solution in strongly polynomial time.

In \Cref{st:xstar}, we calculate an optimal solution to $\PHK$ as defined in \Cref{eq:phk}, while a number of linear programs of form $\LP(a)$ arise throughout the dynamic program in \Cref{st:Bgood}.
The feasibility of the cut constraints can be checked in strongly polynomial time, by solving a minimum cut problem.
The number of degree constraints in both types of these LPs and the number of constraints $x(\delta(B)) \geq 3$ in \Cref{eq:LPa} is polynomial in $n$.
Finally, in \Cref{eq:lp_poly}, the number of constraints involving hyperedges is polynomial in $n$, and one can check the feasibility of the constraints involving the border functions using submodular minimization~\cite{IwataFS2001}.
This means all of the linear programs arising in \Cref{alg:mvsttsp15} can be solved in strongly polynomial time.

According to \Cref{lem:n_cuts_B}, the number of cuts in $\B$ is polynomial in $n$, hence \Cref{st:xstar,st:Bgood,st:polym_alg} can be performed in strongly polynomial time.
This is true for computing a matching in \Cref{st:match}, as well as all the remaining graph operations, using the same arguments as in case of \Cref{alg:apx_tp_path}.
Therefore we provide a $\nicefrac32$-approximation for the metric {\sc Many-visits Path TSP} in strongly polynomial time.

\section{Discussion}
\label{sec:discussion}
In this paper we gave an approximation algorithm for a far-reaching generalization of the metric {\sc Path TSP}, the metric {\sc Many-visits Path TSP} where each city $v$ has a (potentially exponentially large) requirement $r(v)\geq 1$.
Our algorithm yields a $\nicefrac32$-approximation for the metric {\sc Many-visits Path TSP} in time polynomial in the number $n$ of cities and the logarithm of the $r(v)$'s.
It therefore generalizes the recent fundamental result by Zenklusen~\cite{Zenklusen2019}, who obtained a $\nicefrac32$-approximation for the metric {\sc Path TSP}, finishing a long history of research.

At the heart of our algorithm is the first polynomial-time approximation algorithm for the minimum-cost degree bounded g-polymatroid element with multiplicities problem.
That algorithm yields a solution of cost at most the optimum, which violates the lower bounds only by a constant factor depending on the weighted maximum element frequency $\Delta$.

Finally, we show a simple approach, that gives a $\nicefrac52$-approximation for the metric {\sc Many-visits TSP} in strongly polynomial time, and an $\O(1)$-approximation for the metric {\sc Many-visits ATSP} in polynomial time.

\medskip
\noindent
\textbf{Acknowledgements.}
{\small
The authors are grateful to Rico Zenklusen for discussions on techniques to obtain a $\nicefrac32$-approximation for the metric version of the {\sc Many-visits TSP}, and to Tam\'as Kir\'aly and Gyula Pap for their suggestions. Krist\'of B\'erczi was supported by the J\'anos Bolyai Research Fellowship of the Hungarian Academy of Sciences and by the \'UNKP-19-4 New National Excellence Program of the Ministry for Innovation and Technology. Projects no. NKFI-128673 and no. ED\_18-1-2019-0030 (Application-specific highly reliable IT solutions) have been implemented with the support provided from the National Research, Development and Innovation Fund of Hungary, financed under the FK\_18 and the Thematic Excellence Programme funding schemes, respectively.}

\bibliographystyle{abbrv}
\bibliography{mvtsp_apx_new} 

\end{document}